\documentclass[11pt]{article} 

\usepackage[english]{babel}
\usepackage[T1]{fontenc} 
\usepackage[utf8]{inputenc}

\usepackage{color}
\usepackage{amsthm,amsmath,amssymb}
\allowdisplaybreaks

\usepackage{vmargin}
\setmarginsrb{3cm}{2cm}{3cm}{2cm}{0cm}{2cm}{0cm}{1cm}

\usepackage{rotating}
\usepackage{graphicx}
\usepackage{tikz}
\usetikzlibrary{arrows, automata, shapes, positioning, decorations}

\definecolor{bientotlafin}{RGB}{142, 162, 198}
\definecolor{green}{RGB}{31,160,85}
\definecolor{vert}{RGB}{0,255,127}

\theoremstyle{plain}
\newtheorem{theorem}{Theorem}
\newtheorem{lemma}[theorem]{Lemma}
\newtheorem{corollary}[theorem]{Corollary}
\newtheorem{proposition}[theorem]{Proposition}

\theoremstyle{definition}
\newtheorem{definition}[theorem]{Definition}
\newtheorem{example}[theorem]{Example}
\newtheorem{remark}[theorem]{Remark} 

\newcommand{\N}{\mathbb N}
\newcommand{\Z}{\mathbb Z}
\newcommand{\rep}{\mathrm{rep}}
\newcommand{\val}{\mathrm{val}}
\newcommand{\card}{\mathrm{Card}}
\newcommand{\A}{\mathcal{A}}
\newcommand{\T}{\mathcal{T}}
\newcommand{\DIV}{\mathrm{DIV}}
\newcommand{\MOD}{\mathrm{MOD}}
\newcommand{\andrm}{\ {\rm and}\ }
\newcommand{\llfloor}{\big\lfloor}
\newcommand{\rrfloor}{\big\rfloor}
\newcommand{\llceil}{\big\lceil}
\newcommand{\rrceil}{\big\rceil}

\title{Minimal automaton for multiplying and translating the Thue-Morse set}

\author{Émilie Charlier, Célia Cisternino, Adeline Massuir}

\begin{document}
\maketitle

\begin{abstract}
The Thue-Morse set $\mathcal{T}$ is the set of those non-negative integers whose binary expansions have an even number of $1$.  The name of this set comes from the fact that its characteristic sequence is given by the famous Thue-Morse word ${\tt abbabaabbaababba\cdots}$, which is the fixed point starting with ${\tt a}$ of the word morphism ${\tt a\mapsto ab,b\mapsto ba}$. The numbers in $\mathcal{T}$ are commonly called the {\em evil numbers}. We obtain an exact formula for the state complexity of the set $m\T+r$ (i.e.\ the number of states of its minimal automaton) with respect to any base $b$ which is a power of $2$. Our proof is constructive and we are able to explicitly provide the minimal automaton of the language of all $2^p$-expansions of the set of integers $m\T+r$ for any positive integers $p$ and $m$ and any remainder $r\in\{0,\ldots,m-1\}$. The proposed method is general for any $b$-recognizable set of integers. As an application, we obtain a decision procedure running in quadratic time for the problem of deciding whether a given $2^p$-recognizable set is equal to a set of the form $m\T+r$.
\end{abstract}

\section{Introduction}

A subset $X$ of $\N$ is said to be {\em $b$-recognizable} if the base-$b$ expansions of the elements of $X$ form a regular language. The famous theorem of Cobham tells us that any non-trivial property of numbers are dependent on the base we choose: the only sets that are $b$-recognizable for all bases $b$ are the finite unions of arithmetic progressions \cite{Cobham:1969}. Inspired by this seminal result, many descriptions of $b$-recognizable sets were given, e.g.\ morphic, algebraic and logical characterizations \cite{Boigelot&Rassart&Wolper:1998,Bruyere&Hansel&Michaux&Villemaire:1994,Cobham:1972}, extensions of these to systems based on a Pisot number \cite{Bruyere&Hansel:1997}, the normalization map \cite{Frougny:1992} or the possible growth functions \cite{Charlier&Rampersad:2011,Eilenberg:1974}. For more on $b$-recognizable sets, we refer to the surveys \cite{Allouche&Shallit:2003,Bruyere&Hansel&Michaux&Villemaire:1994,Charlier:2018,Eilenberg:1974,Frougny&Sakarovitch:2010,Rigo:2014}.

In particular, as mentioned above, these sets have been characterized in terms of logic. More precisely, a subset of $\N$ (and more generally of $\N^d$) is $b$-recognizable if and only if it is definable by a first-order formula of the structure $\langle\N,+,V_b\rangle$ where $V_b$ is the base-dependent functional predicate that associates with a natural $n$ the highest power of $b$ dividing $n$. Since the finite unions of arithmetic progressions are precisely the subsets of $\N$ that are definable by first order formulas in the Presburger arithmetic $\langle\mathbb{N},+\rangle$, this characterization provides us with a logical interpretation of Cobham’s theorem. In addition, this result turned out to be a powerful tool for showing that many properties of $b$-automatic sequences are decidable and, further, that many enumeration problems of $b$-automatic sequences can be described by $b$-regular sequences in the sense of Allouche and Shallit \cite{Allouche&Shallit:1992,Allouche&Shallit:2003,Charlier&Rampersad&Shallit:2012}. 

In the context of Cobham's theorem, the following question is natural 
and has received a constant attention during the last 30 years: given an automaton accepting the language of the base-$b$ expansions of a set $X\subseteq\N$, is it decidable whether $X$ is  a finite union of arithmetic progressions? Several authors gave decision procedures for this problem \cite{Allouche&Rampersad&Shallit:2009,Bruyere&Hansel&Michaux&Villemaire:1994,Honkala:1986,Leroux:2005,Marsault&Sakarovitch:2013}. Moreover, a multidimensional version of this problem was shown to be decidable in a beautiful way based on logical methods \cite{Bruyere&Hansel&Michaux&Villemaire:1994,Muchnik:2003}. 

With any set of integers $X$ is naturally associated an infinite word, which is its characteristic sequence $\chi_X\colon n\mapsto 1$ if $n\in X,\  n\mapsto 0$ otherwise. Thus, to a finite union of arithmetic progressions corresponds an ultimately periodic infinite word. Therefore, the HD0L ultimate periodicity problem consisting in deciding whether a given morphic word (i.e.\ the image under a coding of the fixed point of a morphism) is ultimately periodic is a generalization of the periodicity problem for $b$-recognizable sets mentioned in the previous paragraph. The HD0L ultimate periodicity problem
was shown to be decidable in its full generality \cite{Durand:2013,Mitrofanov:2013}. The proofs rely on return words, primitive substitutions or evolution of Rauzy graphs. However, these methods do not provide algorithms that could be easily implemented and the corresponding time complexity is very high. In addition, they do not allow us to obtain an algorithm for the multidimensional generalization of the periodicity problem, i.e.\ the problem of deciding whether a $b$-recognizable subset of $\N^d$ is definable within the Presburger arithmetic $\langle\N,+\rangle$. Therefore, a better understanding of the inner structure of automata arising from number systems remains a powerful tool to obtain efficient decision procedures. 

The general idea is as follows. Suppose that $\mathcal{L}=\{L_i\colon i\in \N\}$ is a collection of languages and that we want to decide whether some particular language $L$ belongs to $\mathcal{L}$. Now, suppose that we are able to explicitly give a lower bound on the state complexities of the languages in $\mathcal{L}$, i.e.\ for each given $N$, we can effectively produce a bound $B(N)$ such that for all $i> B(N)$, the state complexity of $L_i$ is greater than $N$. Then the announced problem is decidable: if $N$ is the state complexity of the given language $L$, then only the finitely many languages $L_0,\ldots,L_{B(N)}$ have to be compared with $L$. 

The state complexity of a $b$-recognizable set (i.e.\ the number of states of the minimal automaton accepting the $b$-expansions of its elements) is closely related to the length of the logical formula describing this set. Short formulas are crucial in order to produce efficient mechanical proofs by using for example the Walnut software \cite{Mousavi:2015,Shallit:2015}. There are several ways to improve the previous decision procedure. One of them is to use precise knowledge of the stucture of the involved automata. This idea was successfully used in the papers \cite{BMMR,Marsault&Sakarovitch:2013}. In \cite{Charlier&Rampersad&Rigo&Waxweiler:2011}, the structure of automata accepting the greedy expansions of $m\N$ for a wide class of non-standard numeration systems, and in particular, estimations of the state complexity of $m\N$ are given. Another way of improving this procedure is to have at our disposal the exact state complexities of the languages in $\mathcal{L}$. Finding an exact formula is a much more difficult problem than finding good estimates. However, some results in this direction are known. For instance, it is proved in \cite{Charlier&Rampersad&Rigo&Waxweiler:2011} that for the Zeckendorf numeration system (i.e.\ based on the Fibonacci numbers), the state complexity of $m\N$ is exactly $2m^2$. A complete description of the minimal automaton recognizing $m\N$ in any integer base $b$ was given in~\cite{Alexeev:2004} and the state complexity of $m\N$ with respect to the base $b$ is shown to be exactly 
\[
	\frac{m}{\gcd(m,b^N)}+\sum_{t=0}^{N-1} \frac{b^t}{\gcd(m,b^t)}
\]
where $N$ is the smallest integer $\alpha$ such that $\frac{m-b^\alpha}{\gcd(m,b^\alpha)}<\frac{m}{\gcd(m,b^{\alpha+1})}$. 

For all the above mentioned reasons, the study of the state complexity of $b$-recognizable sets deserves special interest. In the present work, we propose ourselves to initiate a study of the state complexity of sets of the form $mX+r$, for any recognizable subset $X$ of $\N$ (with respect to a given numeration system), any multiple $m$ and any remainder $r$. In doing so, we aim at generalizing the previous framework concerning the case $X=\N$ only.  Our study starts with the Thue-Morse set $\T$ of the so-called {\em evil numbers} \cite{Allouche:2015}, i.e.\ the natural numbers whose base-$2$ expansions contain an even number of occurrences of the digit $1$. The characteristic sequence of this set corresponds to the ubiquitous Thue-Morse word ${\tt abbabaabbaababba\cdots}$, which is the fixed point starting with ${\tt a}$ of the morphism ${\tt a\mapsto ab,b\mapsto ba}$. This infinite word is one of the archetypical aperiodic automatic words, see the surveys \cite{Allouche&Shallit:2009,Queffelec:2018}. Many number-theoretic works devoted to sets of integers defined thanks to the Thue-Morse word exist, such as the study of additive and multiplicative properties, or iterations and sums of such sets \cite{Allouche&Cloitre&Shevelev:2016,Bucci&Hindman&Puzynina&Zamboni:2013,Mauduit:2001}. In this vein, the set $\T$ seems to be a natural candidate to start with. The goal of this work is to provide a complete characterization of the minimal automata recognizing the sets $m\T+r$ for any multiple $m$ and remainder $r$, and any base $b$ which is a power of $2$ (other bases are not relevant with the choice of the Thue-Morse set in view of Cobham's theorem). A previous work dealing with the case $r=0$ recently appeared as a conference paper \cite{Charlier&Cisternino&Massuir:2019}. Surprisingly, the techniques of the present work, in particular the description of the left quotients (i.e.\ the states of the minimal automaton), are quite different from those we developed for the case $r=0$.

This paper has the following organization. In Section~\ref{sec:basics}, we recall the background that is necessary to tackle our problem. In Section~\ref{sec:method}, we state our main result and expose the method that will be carried out for its proof. More precisely, we present the steps of our construction of the minimal automaton accepting the base-$2^p$ expansions of the elements of $m\T+r$ for any positive integers $p$ and $m$, and any remainder $r\in\{0,\ldots,m-1\}$. In Section~\ref{sec:def-automata}, we give the details of the construction of the intermediate automata. In particular, we study the transitions of each automaton. Thus, at the end of Section~\ref{sec:def-automata}, we are provided with an automaton recognizing the desired language. Then in Section~\ref{sec:properties-automata}, we study the properties of the built automata that will be needed for proving the announced state complexity result. The minimization procedure of the last automaton is handled in Section~\ref{sec:minimization}. This part is the most technical one and it deeply relies on the properties of the intermediate automata proved in the previous sections. In Section~\ref{sec:decision}, we show that as an application of our  results, we obtain a decision procedure running in quadratic time for the problem of deciding whether a given $2^p$-recognizable set is equal to a set of the form $m\T+r$.  In Section~\ref{sec:r=0}, we explicitly give the correspondence between the description of the minimal automaton recognizing $m\T$ obtained in \cite{Charlier&Cisternino&Massuir:2019} and that given in the present work in the particular case where $r=0$. In Section~\ref{sec:Tc}, we show that the minimal automaton recognizing $m\overline{\T}+r$,  where $\overline{\T}$ is the complement of the Thue-Morse set $\T$, is obtained directly from the one recognizing  $m\T+r$ by moving the initial state. As a consequence, the state complexities of $m\overline{\T}+r$ and $m\T+r$ coincide. Finally, in Section~\ref{sec:perspectives}, we discuss future work and give two related open problems.

\section{Basics}
\label{sec:basics}

In this text, we use the usual definitions and notation (alphabet, letter, word, language, free monoid, automaton, etc.) of formal language theory; for example, see \cite{Lothaire1997,Sakarovitch2009}.  

Nevertheless, let us give a few definitions and properties that will be central in this work. The empty word is denoted by $\varepsilon$. For a finite word $w$, $|w|$ designates its length and $|w|_a$ the number of occurrences of the letter $a$ in $w$.  A {\em regular language} is a language which is accepted by a finite automaton. For $L\subseteq A^*$ and $w\in A^*$, the {\em (left) quotient} of $L$ by $w$ is the language
\[
	w^{-1}L=\{u\in A^*\colon wu\in L\}.
\]
As is well known, a language $L$ over an alphabet $A$ is regular if and only if it has finitely many quotients, that is, the set of languages
\[
	\{w^{-1}L\colon w\in A^*\}
\]
is finite. The {\em state complexity} of a regular language is the number of its quotients: $\card(\{w^{-1}L\colon w\in A^*\})$. It corresponds to the number of states of its minimal automaton. The following characterization of minimal automata will be used several times in this work: a deterministic finite automaton (or DFA for short) is minimal if and only if it is complete, reduced and accessible. A DFA is said to be {\em complete} if the transition function is total (i.e.\ from every state start transitions labeled with all possible letters), {\em reduced} if languages accepted from distinct states are distinct and {\em accessible} if every state can be reached from the initial state. The language accepted from a state $q$ is denoted by $L_q$. Thus, the language accepted by a DFA is the language accepted from its initial state.

In what follows we will need a notion that is somewhat stronger than that of reduced DFAs. We say that a DFA has {\em disjoint states} if the languages accepted from distinct states are disjoint: for distinct states $p$ and $q$, we have $L_p\cap L_q=\emptyset$. A state $q$ is said to be {\em coaccessible} if $L_q\ne \emptyset$ and, by extension, an automaton is {\em coaccessible} if all its states are coaccessible. Thus, any coaccessible DFA having disjoint states is reduced.

Now, let us give some background on numeration systems. Let $b\in\N_{\ge2}$. We define $A_b$ to be the alphabet $\{\tt{0},\ldots,\tt{b{-}1}\}$. Elements of $A_b$ are called {\em digits}. The number $b$ is called the {\em base} of the numeration. In what follows we will make no distinction between a digit ${\tt c}$ in $A_b$ and its {\em value} $c$ in $[\![0,b{-}1]\!]$. Otherwise stated, we identify the alphabet $A_b$ and the interval of integers $[\![0,b{-}1]\!]$. Note that here and throughout the text, we use the notation $[\![m,n]\!]$ to designate the interval of integers $\{m,m+1,\ldots,n\}$. The {\em $b$-expansion} of a positive integer $n$, which is denoted by $\rep_b(n)$, is the finite word $c_\ell\cdots c_0$ over $A_b$ defined by
\[
	n=\sum_{j=0}^\ell c_j b^j, \quad c_\ell\ne 0.
\]
The {\em  $b$-expansion} of $0$ is the empty word: $\rep_b(0)=\varepsilon$. Conversely, for a word $w=c_\ell\cdots c_0$ over $A_b$, we write $\val_b(w)=\sum_{j=0}^{\ell} c_j b^j$. Thus we have $\rep_b\colon\N\to A_b^*$ and $\val_b\colon A_b^*\to \N$. Clearly, the function $\val_b\circ \rep_b$ is the identity from $\N$ to $\N$. Moreover, for any $w\in A_b^*$, the words $\rep_b(\val_b(w))$ and $w$ only differ by the potential leading zeroes in $w$. Also note that for all subsets $X$ of $\N$, we have $\val_b^{-1}(X)=0^*\rep_b(X)$. A subset $X$ of $\N$ is said to be {\em $b$-recognizable} if the language $\rep_b (X)$ is regular, or equivalently, if the language $\val_b^{-1}(X)$ is regular. In what follows, we will always consider automata accepting $\val_b^{-1}(X)$ instead of $\rep_b(X)$. The {\em state complexity} of a $b$-recognizable subset $X$ of $\N$ {\em with respect to the base $b$} is the state complexity of the language $\val_b^{-1}(X)$. Note that the $b$-expansions are read from left to right, i.e.\ most significant digit first.

We will need to represent not only natural numbers, but also pairs of natural numbers. If $u=u_1\cdots u_n\in A^*$ and $v=v_1\cdots v_n\in B^*$ are words of the same length $n$, then we use the notation $(u,v)$ to designate the word $(u_1,v_1)\cdots (u_n,v_n)$ of length $n$ over the alphabet $A\times B$: 
\[
	(u,v)=(u_1,v_1)\cdots (u_n,v_n)\in (A\times B)^*.
\]
For $(m,n)\in\N^2$, we write
\[
	\rep_b(m,n)=(0^{\ell-|\rep_b(m)|}\rep_b(m),0^{\ell-|\rep_b(n)|}\rep_b(n))
\] 
where $\ell=\max\{|\rep_b(m)|,|\rep_b(n)|\}$. Otherwise stated, we add leading zeroes to the shortest expansion (if any) in order to obtain two words of the same length. Finally, for a subset $X$ of $\N^2$, we write
\[
	\val_{b}^{-1}(X)=(0,0)^*\rep_b(X).
\]

\section{Method}
\label{sec:method}

The Thue-Morse set, which we denote by $\T$, is the set of all natural numbers whose base-$2$ expansions contain an even number of occurrences of the digit $1$:
\[
	\T= \{n\in\N\colon|\rep_2 (n)|_1\in 2\N\}.
\]
The Thue-Morse set $\T$ is $2$-recognizable since the language $\val_2^{-1}(\T)$ is accepted by the automaton depicted in Figure~\ref{fig:aut-TM-2}.
\begin{figure}[htb]
\centering
\begin{tikzpicture}
\tikzstyle{every node}=[shape=circle, fill=none, draw=black,minimum size=20pt, inner sep=2pt]
\node(1) at (0,0) {$T$};
\node(2) at (0,-1.5) {$B$};
\tikzstyle{every node}=[shape=circle, fill=none, draw=black,minimum size=15pt, inner sep=2pt]
\node(2f) at (0,0) {};
\tikzstyle{every path}=[color=black, line width=0.6 pt]
\tikzstyle{every node}=[shape=circle, minimum size=5pt, inner sep=2pt]
\draw [->] (-1,0) to node {} (1); 
\draw [->] (1) to [loop above] node [above] {$0$} (1);
\draw [->] (2) to [loop below] node [below] {$0$} (2);
\draw [->] (1) to [bend right=30] node [left] {$1$} (2);
\draw [->] (2) to [bend right=30] node [right] {$1$} (1);
\end{tikzpicture}
\caption{The Thue-Morse set is $2$-recognizable.}
\label{fig:aut-TM-2}
\end{figure}
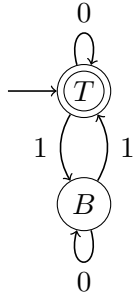
More precisely, the Thue-Morse set $\T$ is $2^p$-recognizable for all $p\in\N_{\ge1}$ and is not $b$-recognizable for any other base $b$. This is a consequence of the famous theorem of Cobham. Two positive integers are said to be {\em multiplicatively independent} if their only common integer power is $1$ and are said {\em multiplicatively dependent} otherwise. 

\begin{theorem}[Cobham \cite{Cobham:1969}] \ 
\begin{itemize}
\item Let $b,b'$ be two multiplicatively independent bases. Then a subset of $\N$ is both $b$-recognizable and $b'$-recognizable if and only if it is a finite union of arithmetic progressions.
\item Let $b,b'$ be two multiplicatively dependent bases. Then a subset of $\N$ is $b$-recognizable if and only if it is $b'$-recognizable.
\end{itemize}
\end{theorem}

In the case of the Thue-Morse set, it is easily seen that, for each $p\in\N_{\ge1}$, the language $\val_{2^p}^{-1}(\T)$ is accepted by the DFA $(\{T,B\},T,T,A_{2^p},\delta)$ where for all $X\in\{T,B\}$ and all $a\in A_{2^p}$,  
\[
	\delta(X,a)=\begin{cases}
					X & \text{if } a\in \T \\
					\overline{X} & \text{else}
				\end{cases}
\]
where $\overline{T}=B$ and $\overline{B}=T$. For example this automaton is depicted in Figure~\ref{fig:aut-TM-4} for $p=2$.
\begin{figure}[htb]
\centering
\begin{tikzpicture}
\tikzstyle{every node}=[shape=circle, fill=none, draw=black,
minimum size=20pt, inner sep=2pt]
\node(1) at (0,0) {$T$};
\node(2) at (0,-1.5) {$B$};
\tikzstyle{every node}=[shape=circle, fill=none, draw=black,
minimum size=15pt, inner sep=2pt]
\node(2f) at (0,0) {};
\tikzstyle{every path}=[color=black, line width=0.5 pt]
\tikzstyle{every node}=[shape=circle, minimum size=5pt, inner sep=2pt]
\draw [->] (-1,0) to node {} (1); 
\draw [->] (1) to [loop above] node [above] {$0,3$} (1);
\draw [->] (2) to [loop below] node [below] {$0,3$} (2);
\draw [->] (1) to [bend right=30] node [left] {$1,2$} (2);
\draw [->] (2) to [bend right=30] node [right] {$1,2$} (1);
\end{tikzpicture}
\caption{The Thue-Morse set is $4$-recognizable.}
\label{fig:aut-TM-4}
\end{figure}
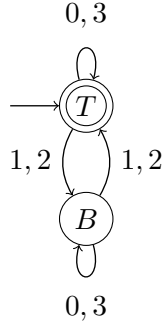

In order to avoid a systematic case separation, we introduce the following notation: for $X\in\{T,B\}$ and $n \in \N$, we define
\[
	X_n = \begin{cases}
			X & \text{if } n \in \T \\
			\overline{X} & \text{else.}  
			\end{cases}
\]
With this notation, we can simply rewrite the definition of the transition function $\delta$ as $\delta(X,a)=X_a$.

The following proposition is well known; for example see~\cite{Bruyere&Hansel&Michaux&Villemaire:1994}. 

\begin{proposition}
Let $b\in\N_{\ge 2}$ and $m,t\in\N$. If $X$ is $b$-recognizable, then so is $mX+t$. Otherwise stated, the dilation $n\mapsto mn$ and translation $n\mapsto n+t$ preserve $b$-recognizability.
\end{proposition} 

In particular, for any $m,t\in\N$ and $p\in\N_{\ge 1}$, the set $m\T+t$ is $2^p$-recognizable. The aim of this work is to show the following result.

\begin{theorem} 
\label{thm:main}
Let $m,p$ be positive integers and $r\in[\![0,m-1]\!]$. Then the state complexity of $m\T+r$ with respect to the base $2^p$ is equal to 
\[
	2k+\left\lceil \frac zp\right\rceil
\]
if $m=k2^z$ with $k$ odd.
\end{theorem}

Our proof of Theorem~\ref{thm:main} is constructive. In order to describe the minimal DFA of $\val_{2^p}^{-1}(m\T+r)$, we will successively construct  several automata. First, we build a DFA $\A_{\T,2^p}$ accepting the language 
\[
	\val_{2^p}^{-1}(\T\times \N).
\]
Then we build a DFA $\A_{m,r,b}$ accepting the language
\[
	\val_b^{-1}\big(\{(n,mn+r)\colon n\in \N\}\big).
\] 
Note that we do the latter step for any integer base $b$ and not only for powers of $2$.
Next, we consider the product automaton $\A_{m,r,2^p}\times\A_{\T,2^p}$. This DFA accepts the language 
\[
	\val_{2^p}^{-1}\big(\{(t,mt+r)\colon t \in \T\}\big).
\]
Finally, a finite automaton $\Pi(\A_{m,r,2^p}\times\A_{\T,2^p})$ accepting $\val_{2^p}^{-1}(m\T+r)$ is obtained by projecting the label of each transition in $\A_{m,r,2^p}\times\A_{\T,2^p}$ onto its second component.
At each step of our construction, we check that the automaton under consideration is minimal (and hence deterministic) and the ultimate step precisely consists in a minimization procedure.

From now on, we fix some positive integers $m,p$ and some remainder $r\in[\![0,m-1]\!]$. We also let $z$ and $k$ be the unique integers such that $m=k2^z$ with $k$ odd. Finally we let $R=|\rep_{2^p}(r)|$.

\section{Construction of the intermediate automata}
\label{sec:def-automata}

\subsection{The automaton $\A_{\T,2^p}$}

First, we build a DFA  $\A_{\T,2^p}$ accepting the language $\val_{2^p}^{-1}(\T\times \N)$. This DFA is a modified version of the automaton accepting $\val_{2^p}^{-1}(\T)$ defined in the previous section. Namely, we replace each transition labeled by $a\in A_{2^p}$ by $2^p$ copies of itself labeled by $(a,b)$, for each $b\in A_{2^p}$. Formally, 
\[
	\A_{\T,2^p}=(\{T,B\},T,T,A_{2^p}\times A_{2^p},\delta_{\T,2^p})
\]
where, for all $X\in\{T,B\}$ and all $a,b\in A_{2^p}$, we have $\delta_{\T,2^p}(X,(a,b))=X_a$. For example, the automata $\A_{\T,2}$ and $\A_{\T,4}$ are depicted in Figure~\ref{fig:aut-TM-2-4-couples}.
\begin{figure}[htb]
\vspace{-1cm}
\begin{minipage}[c]{0.34\linewidth}
\begin{tikzpicture}
\tikzstyle{every node}=[shape=circle, fill=none, draw=black,
minimum size=20pt, inner sep=2pt]
\node(1) at (0,0) {$T$};
\node(2) at (0,-1.75) {$B$};
\tikzstyle{every node}=[shape=circle, fill=none, draw=black,
minimum size=15pt, inner sep=2pt]
\node(2f) at (0,0) {};
\tikzstyle{every path}=[color=black, line width=0.5 pt]
\tikzstyle{every node}=[shape=circle, minimum size=5pt, inner sep=2pt]
\draw [->] (-1,0) to node {} (1); 
\draw [->] (1) to [loop above] node [above=-0.3cm] {\begin{tabular}{c}
$(0,0)$ \\ 
$(0,1)$ \\
\end{tabular}} (1);
\draw [->] (2) to [loop below] node [below=-0.3cm] {\begin{tabular}{c}
$(0,0)$ \\
$(0,1)$ \\
\end{tabular}} (2);
\draw [->] (1) to [bend right=30] node [left=-0.3cm] {\begin{tabular}{c}
$(1,0)$ \\
$(1,1)$ \\
\end{tabular}} (2);
\draw [->] (2) to [bend right=30] node [right=-0.3cm] {\begin{tabular}{c}
$(1,0)$ \\
$(1,1)$ \\
\end{tabular}} (1);
\end{tikzpicture}
\end{minipage}
\begin{minipage}[c]{0.64\linewidth}
\begin{tikzpicture}
\tikzstyle{every node}=[shape=circle, fill=none, draw=black,
minimum size=20pt, inner sep=2pt]
\node(1) at (0,0) {$T$};
\node(2) at (0,-1.75) {$B$};
\tikzstyle{every node}=[shape=circle, fill=none, draw=black,
minimum size=15pt, inner sep=2pt]
\node(2f) at (0,0) {};
\tikzstyle{every path}=[color=black, line width=0.5 pt]
\tikzstyle{every node}=[shape=circle, minimum size=5pt, inner sep=2pt]
\draw [->] (-1,0) to node {} (1); 
\draw [->] (1) to [loop above] node [above=-1.7cm] {\begin{tabular}{c}
$(0,0),(0,1),(0,2),(0,3)$ \\
$(3,0),(3,1),(3,2),(3,3)$ \\
\end{tabular}} (1);
\draw [->] (2) to [loop below] node [below=-1.7cm] {\begin{tabular}{c}
$(0,0),(0,1),(0,2),(0,3)$ \\
$(3,0),(3,1),(3,2),(3,3)$ \\
\end{tabular}} (2);
\draw [->] (1) to [bend right=30] node [left=-0.2cm] {\begin{tabular}{c}
$(1,0),(1,1),(1,2),(1,3)$ \\
$(2,0),(2,1),(2,2),(2,3)$ \\
\end{tabular}} (2);
\draw [->] (2) to [bend right=30] node [right=-0.2cm] {\begin{tabular}{c}
$(1,0),(1,1),(1,2),(1,3)$ \\
$(2,0),(2,1),(2,2),(2,3)$ \\
\end{tabular}} (1);
\end{tikzpicture}
\end{minipage}
\vspace{-1cm}
\caption{The automata $\A_{\T,2}$ (left) and $\A_{\T,4}$ (right).}
\label{fig:aut-TM-2-4-couples}
\end{figure}
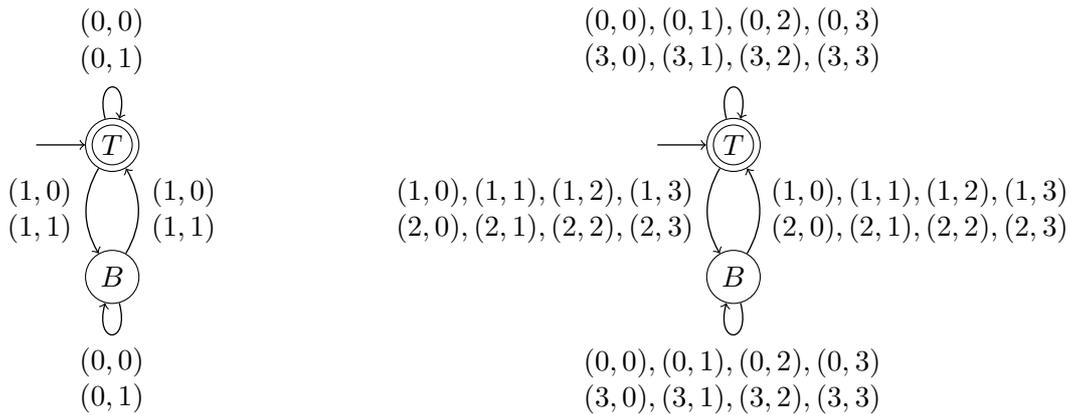

\begin{lemma}
\label{lemlem:transitionsTM}
Let $u,v \in A_{2^p}^*$. Then $\val_{2^p}(uv)\in\T$ if and only if, either $\val_{2^p}(u)\in\T$ and  $\val_{2^p}(v)\in\T$, or $\val_{2^p}(u)\notin\T$ and  $\val_{2^p}(v)\notin\T$.
\end{lemma}

\begin{proof}
Let $\tau\colon A_{2^p}^*\to A_{2^p}^*$ be the $p$-uniform morphism defined by $\tau(a)=0^{p-|\rep_2(a)|}\rep_2(a)$ for each $a\in A_{2^p}$. Then, for all $w\in A_{2^p}^*$, we have $\val_{2^p}(w)=\val_2(\tau(w))$. Therefore, $\val_{2^p}(w)\in\T$ if and only if $|\tau(w)|_1\in 2\N$.
Since $\tau$ is a morphism, we have $|\tau(uv)|_1=|\tau(u)|_1+|\tau(v)|_1$. Hence $|\tau(uv)|_1$ is even if and only if $|\tau(u)|_1$ and $|\tau(v)|_1$ are both even or both odd. 
\end{proof}

\begin{lemma}
\label{lem:transitionsTM}
For all $X \in \{T,B\}$ and $(u,v)\in (A_{2^p}\times A_{2^p})^*$, we have 
\[
	\delta_{\T,2^p}(X,(u,v))=X_{\val_{2^p}(u)}.
\]
\end{lemma}

\begin{proof}
We do the proof by induction on $|(u,v)|$. The case $|(u,v)|=0$ is trivial. 
Now let $X\in \{T,B \}$ and let $(ua,vb)\in (A_{2^p}\times A_{2^p})^*$ with $a,b\in A_{2^p}$. We suppose that the result is satisfied for $(u,v)$ and we show that it is also true for $(ua,vb)$. Let $Y=\delta_{\T,2^p}(X,(u,v))$. By induction hypothesis, we have $Y =X_{\val_{2^p}(u)}$. Thus we obtain
\[
	\delta_{\T,2^p}(X,(ua,vb)) 
	=\delta_{\T,2^p}(Y,(a,b))
	=Y_a
	=(X_{\val_{2^p}(u)})_a
	=X_{\val_{2^p}(ua)}
\]
where we have used Lemma~\ref{lemlem:transitionsTM} for the last step.
\end{proof}

\subsection{The automaton $\A_{m,r,b}$}

In this section, we consider an arbitrary integer base $b$. Let
\[
	\A_{m,r,b}=([\![0,m{-}1]\!],0,r,A_b\times A_b,\delta_{m,r,b})
\]
where the (partial) transition function $\delta_{m,r,b}$ is defined as follows: for $i,j\in[\![0,m{-}1]\!]$ and $d,e\in A_b$, we set
\[
	\delta_{m,r,b}(i,(d,e))=j \quad \iff \quad bi + e = md + j.
\]
The DFA $\A_{m,r,b}$ accepts the language $\val_{b}^{-1}\big(\{(n,mn+r)\colon n\in \N\}\big)$. We refer the interested reader to \cite{Waxweiler2009}. For example, the automaton $\A_{6,2,4}$ is depicted in Figure~\ref{fig:A-6,4}. 

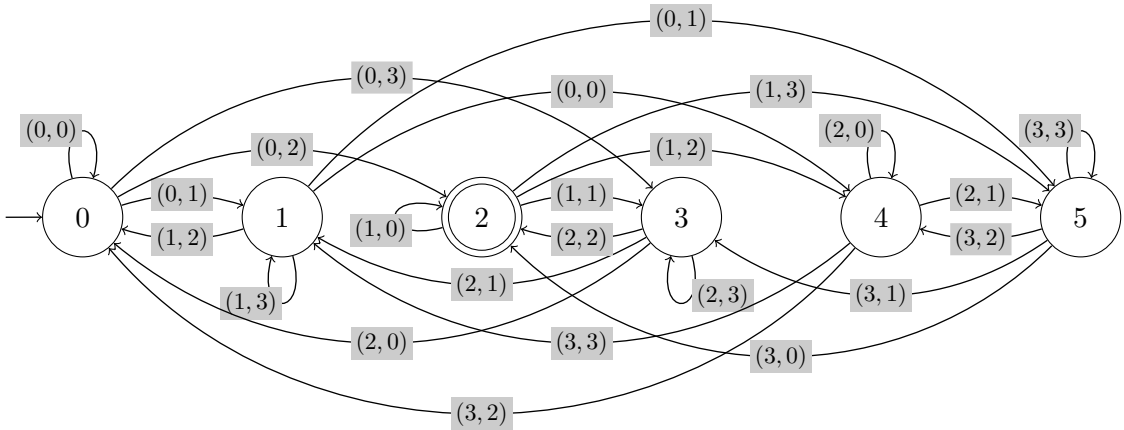
\begin{figure}[htb]
\centering
\begin{tikzpicture}[scale=0.75]
\tikzstyle{every node}=[shape=circle, fill=none, draw=black,
minimum size=30pt, inner sep=2pt]
\node(0) at (0,6.5) {$0$};
\node(1) at (3.5,6.5) {$1$};
\node(2) at (7,6.5) {$2$};
\node(3) at (10.5,6.5) {$3$};
\node(4) at (14,6.5) {$4$};
\node(5) at (17.5,6.5) {$5$};

\tikzstyle{every node}=[shape=circle, fill=none, draw=black,
minimum size=25pt, inner sep=2pt]
\node at (7,6.5) {};

\tikzstyle{etiquettedebut}=[very near start,rectangle,fill=black!20,scale=0.9]
\tikzstyle{etiquettemilieu}=[midway,rectangle,fill=black!20,scale=0.9]
\tikzstyle{every path}=[color=black, line width=0.5 pt,scale=0.9]
\tikzstyle{every node}=[shape=circle, minimum size=5pt, inner sep=2pt,scale=0.9]

\draw [->] (-1.5,7.22) to node {} (0); 
\draw [->] (0) to [loop above] node [left,rectangle,fill=black!20,scale=0.9] {$(0,0)$} (0);
\draw [->] (1) to [loop below] node [left,rectangle,fill=black!20,scale=0.9] {$(1,3)$} (1);
\draw [->] (2) to [loop left] node [left,pos=0.2,rectangle,fill=black!20,scale=0.9] {$(1,0)$} (2);
\draw [->] (3) to [loop below] node [right,pos=0.3,rectangle,fill=black!20,scale=0.9] {$(2,3)$} (3);
\draw [->] (4) to [loop above] node [left,rectangle,fill=black!20,scale=0.9] {$(2,0)$} (4);
\draw [->] (5) to [loop above] node [left,rectangle,fill=black!20,scale=0.9] {$(3,3)$} (5);
\draw [->] (0) to [bend left=17] node [sloped,etiquettemilieu] {$(0,1)$} (1);
\draw [->] (0) to [bend left=30] node [sloped,etiquettemilieu] {$(0,2)$} (2);
\draw [->] (0) to [bend left=45] node [sloped,etiquettemilieu] {$(0,3)$} (3);
\draw [->] (1) to [bend left=40] node [sloped,etiquettemilieu] {$(0,0)$} (4);
\draw [->] (1) to [bend left=50] node [sloped,etiquettemilieu] {$(0,1)$} (5);
\draw [->] (1) to [bend left=17] node [sloped,etiquettemilieu] {$(1,2)$} (0);
\draw [->] (2) to [bend left=17] node [sloped,etiquettemilieu] {$(1,1)$} (3);
\draw [->] (2) to [bend left=30] node [sloped,etiquettemilieu] {$(1,2)$} (4);
\draw [->] (2) to [bend left=40] node [sloped,etiquettemilieu] {$(1,3)$} (5);
\draw [->] (3) to [bend left=40] node [sloped,etiquettemilieu] {$(2,0)$} (0);
\draw [->] (3) to [bend left=30] node [sloped,etiquettemilieu] {$(2,1)$} (1);
\draw [->] (3) to [bend left=17] node [sloped,etiquettemilieu] {$(2,2)$} (2);
\draw [->] (4) to [bend left=17] node [sloped,etiquettemilieu] {$(2,1)$} (5);
\draw [->] (4) to [bend left=50] node [sloped,etiquettemilieu] {$(3,2)$} (0);
\draw [->] (4) to [bend left=40] node [sloped,etiquettemilieu] {$(3,3)$} (1);
\draw [->] (5) to [bend left=45] node [sloped,etiquettemilieu] {$(3,0)$} (2);
\draw [->] (5) to [bend left=35] node [sloped,etiquettemilieu] {$(3,1)$} (3);
\draw [->] (5) to [bend left=17] node [sloped,etiquettemilieu] {$(3,2)$} (4);
\end{tikzpicture}
\caption{The automaton $\A_{6,2,4}$ accepts the language $\val_4^{-1}\big(\{(n,6n+2)\colon n\in \N\}\big)$.}
\label{fig:A-6,4}
\end{figure}

Note that the automaton $\A_{m,r,b}$ is not complete (see Remark~\ref{rem:unicitelettrepremcomp}). Also note that there is always a loop labeled by $(0,0)$ on the initial state $0$.

\begin{remark}
\label{rem:unicitelettrepremcomp}
For each $i \in [\![0,m{-}1]\!]$ and $e \in A_b$, there exist unique $d \in A_b$ and $j \in [\![0,m{-}1]\!]$ such that
$\delta_{m,r,b}(i,(d,e))=j$. Indeed, $d$ and $j$ are unique since they are the quotient and remainder of the Euclidean division of $bi+e$ by $m$. We still have to check that $d < b$. We have
\[
	bi + e = md + j \iff d = \frac{bi+e-j}{m}.
\]
Since $i\le m{-}1$, $j \ge 0$ and $e < b$, we have
\[
   	\frac{bi+e-j}{m}< \frac{b(m{-}1)+b}{m}=b.
\]
\end{remark}

\begin{lemma} 
\label{lem:transitionsAmb}
For $i,j\in [\![0,m{-}1]\!]$ and $(u,v) \in (A_b\times A_b)^*$, we have
\[
	\delta_{m,r,b}(i,(u,v))=j \iff b^{|(u,v)|}\, i + \val_b(v) = m\, \val_b(u) + j.
\]
\end{lemma}

\begin{proof}
We do the proof by induction on $n=|(u,v)|$. If $n$ is equal to $0$, the result is clear. Now let $i,j\in [\![0,m{-}1]\!]$ and let $(du,ev)\in (A_b\times A_b)^*$ with $d,e\in A_b$ and $|(u,v)|=n$. We suppose that the result is satisfied for $(u,v)$ and we show that it is also true for $(du,ev)$. We use the notation $\DIV(x,y)$ and $\MOD(x,y)$ to designate the quotient and the remainder of the Euclidean division of $x$ by $y$ (thus, we have $\DIV(x,y)=\llfloor \frac xy \rrfloor$). 
By definition of the transition function, we have
\[
	\delta_{m,r,b}(i,(du,ev))=j
	 \iff  d=\DIV(bi+e,m) \ \andrm \ \delta_{m,r,b}(\MOD(bi+e,m),(u,v))=j.
\]
By using the induction hypothesis, we have
\begin{align*}
	\delta_{m,r,b}(bi+e-md,(u,v))=j
	& \iff b^n\, (bi+e-md) + \val_b(v)= m\, \val_b(u)+j 		\\
	& \iff b^{n+1}\, i + \val_b(ev)= m\, \val_b(du)+j.
\end{align*}
To be able to conclude the proof, we still have to show that 
\begin{equation}
\label{eq:val}
	b^{n+1}\, i + \val_b(ev)= m\, \val_b(du)+j 
\end{equation}
implies 
\[
	d=\DIV(bi+e,m).
\]
Thus, suppose that \eqref{eq:val} is true. Then
\[
	b^{n+1}\, i+ b^n e + \val_b(v)= m(b^nd+ \val_b(u))+j.
\]
Since $\val_b(u)$ and $\val_b(v)$ are less than $b^n$, $d\ge 0$, $j<m$ and $b^nd+ \val_b(u)\ge 0$, 
we obtain
\begin{align*}
	d 	&= \DIV(b^nd+ \val_b(u),b^n) \\
		&= \DIV(\DIV(b^{n+1}\, i+ b^n e + \val_b(v),m),b^n) \\
		&= \DIV(\DIV(b^{n+1}\, i+ b^n e + \val_b(v),b^n),m) \\
		&= \DIV(b\, i+ e,m)
\end{align*}
as desired.
\end{proof}

\begin{remark} 
\label{rem:unicitepremcomp}
It is easily checked that Remark~\ref{rem:unicitelettrepremcomp} extends from letters to words: for each $i\in[\![0,m{-}1]\!]$ and $v\in A_b^*$, there exist unique $u\in A_b^*$ and $j\in[\![0,m{-}1]\!]$ such that $\delta_{m,r,b}(i,(u,v))=j$. In particular, the word $u$ must have the same length as the word $v$, and hence $\val_b(u)<b^{|v|}$.
\end{remark}

\subsection{The projected automaton $\Pi(\A_{m,r,b})$}

In this section again, $b$ is an arbitrary integer base. We consider the automaton obtained by projecting the label of each transition of $\A_{m,r,b}$ onto its second component. We denote by $\Pi(\A_{m,r,b})$ the automaton obtained thanks to this projection. Thanks to Remark~\ref{rem:unicitelettrepremcomp}, the automaton $\Pi(\A_{m,r,b})$ is deterministic and complete. We denote by $\delta_{m,r,b}^\Pi$ the corresponding transition function. For example, the automaton $\Pi(\A_{6,2,4})$ is depicted in Figure~\ref{projA-6,4}.
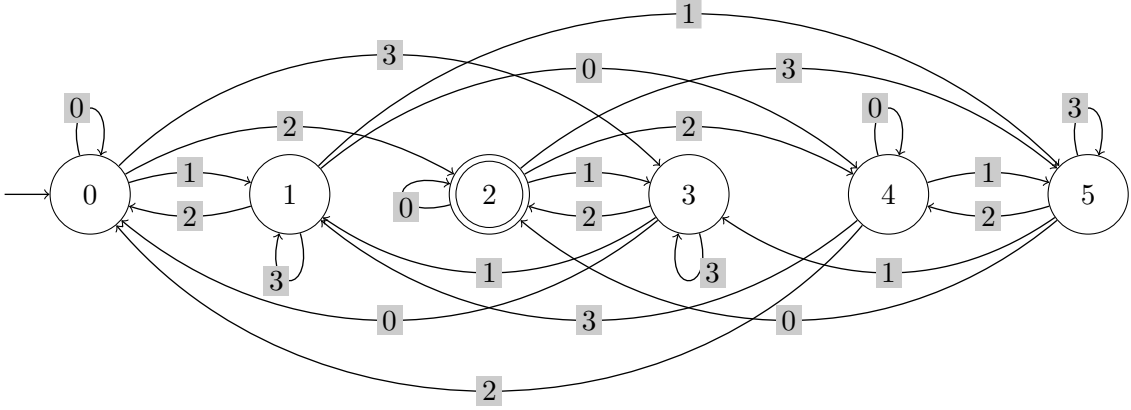
\begin{figure}[htb]
\centering
\begin{tikzpicture}[scale=0.75]
\tikzstyle{every node}=[shape=circle, fill=none, draw=black,
minimum size=30pt, inner sep=2pt]
\node(0) at (0,6.5) {$0$};
\node(1) at (3.5,6.5) {$1$};
\node(2) at (7,6.5) {$2$};
\node(3) at (10.5,6.5) {$3$};
\node(4) at (14,6.5) {$4$};
\node(5) at (17.5,6.5) {$5$};

\tikzstyle{every node}=[shape=circle, fill=none, draw=black,
minimum size=25pt, inner sep=2pt]
\node at (7,6.5) {};

\tikzstyle{etiquettedebut}=[very near start,rectangle,fill=black!20]
\tikzstyle{etiquettemilieu}=[midway,rectangle,fill=black!20]
\tikzstyle{every path}=[color=black, line width=0.5 pt]
\tikzstyle{every node}=[shape=circle, minimum size=5pt, inner sep=2pt]

\draw [->] (-1.5,6.5) to node {} (0); 
\draw [->] (0) to [loop above] node [left,rectangle,fill=black!20] {$0$} (0);
\draw [->] (1) to [loop below] node [left,rectangle,fill=black!20] {$3$} (1);
\draw [->] (2) to [loop left] node [left,pos=0.2,rectangle,fill=black!20] {$0$} (2);
\draw [->] (3) to [loop below] node [right,pos=0.3,rectangle,fill=black!20] {$3$} (3);
\draw [->] (4) to [loop above] node [left,rectangle,fill=black!20] {$0$} (4);
\draw [->] (5) to [loop above] node [left,rectangle,fill=black!20] {$3$} (5);
\draw [->] (0) to [bend left=17] node [sloped,etiquettemilieu] {$1$} (1);
\draw [->] (0) to [bend left=30] node [sloped,etiquettemilieu] {$2$} (2);
\draw [->] (0) to [bend left=45] node [sloped,etiquettemilieu] {$3$} (3);
\draw [->] (1) to [bend left=40] node [sloped,etiquettemilieu] {$0$} (4);
\draw [->] (1) to [bend left=45] node [sloped,etiquettemilieu] {$1$} (5);
\draw [->] (1) to [bend left=17] node [sloped,etiquettemilieu] {$2$} (0);
\draw [->] (2) to [bend left=17] node [sloped,etiquettemilieu] {$1$} (3);
\draw [->] (2) to [bend left=30] node [sloped,etiquettemilieu] {$2$} (4);
\draw [->] (2) to [bend left=40] node [sloped,etiquettemilieu] {$3$} (5);
\draw [->] (3) to [bend left=40] node [sloped,etiquettemilieu] {$0$} (0);
\draw [->] (3) to [bend left=35] node [sloped,etiquettemilieu] {$1$} (1);
\draw [->] (3) to [bend left=17] node [sloped,etiquettemilieu] {$2$} (2);
\draw [->] (4) to [bend left=17] node [sloped,etiquettemilieu] {$1$} (5);
\draw [->] (4) to [bend left=50] node [sloped,etiquettemilieu] {$2$} (0);
\draw [->] (4) to [bend left=40] node [sloped,etiquettemilieu] {$3$} (1);
\draw [->] (5) to [bend left=40] node [sloped,etiquettemilieu] {$0$} (2);
\draw [->] (5) to [bend left=35] node [sloped,etiquettemilieu] {$1$} (3);
\draw [->] (5) to [bend left=17] node [sloped,etiquettemilieu] {$2$} (4);
\end{tikzpicture}
\caption{The projected automaton $\Pi(\A_{6,2,4})$.}
\label{projA-6,4}
\end{figure}

\begin{lemma} 
\label{lem:transitionsPiAmb}
For $i,j\in [\![0,m{-}1]\!]$ and $v \in A_b^*$, we have
\[
	\delta_{m,r,b}^\Pi(i,v)=j \iff 
	b^{|v|}\, i + \val_b(v) \equiv  j\pmod m.
\]
\end{lemma}

\begin{proof}
Let  $i,j\in [\![0,m{-}1]\!]$ and $v \in A_b^*$. If $\delta_{m,r,b}^\Pi(i,v)=j$, then there exists a word $u$ of the same length as $v$ such that $\delta_{m,r,b}(i,(u,v))=j$. By Lemma~\ref{lem:transitionsAmb}, we get $b^{|v|}\, i + \val_b(v) \equiv  j\pmod m$. Conversely, suppose that there exists some $\ell\in\N$ such that $b^{|v|}\, i + \val_b(v) = m\ell + j$. Since $i\le m-1$, $\val_b(v)<b^{|v|}$ and $j\ge0$, we necessarily have 
\[
	\ell=\frac{b^{|v|}i+\val_b(v)-j}{m}<\frac{b^{|v|}(m-1)+b^{|v|}}{m}=b^{|v|}.
\] 
Hence $|\rep_b(\ell)|\le |v|$ and the word $u=0^{|v|-|\rep_b(\ell)|}\rep_b(\ell)$ has length $|v|$ and is such that $\val_b(u)=\ell$. The conclusion follows from Lemma~\ref{lem:transitionsAmb}.
\end{proof}

\subsection{The product automaton $\A_{m,r,2^p}\times \A_{\T,2^p}$}

In this section, we study the product automaton $\A_{m,r,2^p} \times \A_{\T,2^p}$. Since the states of $\A_{m,r,2^p}$ are numbered from $0$ to $m{-}1$ and those of $\A_{\T,2^p}$ are $T$ and $B$, we denote the states of the product automaton by
\[
	(0,T),\ldots,(m{-}1,T) \andrm (0,B) , \ldots , (m{-}1, B).
\]
The transitions of $\A_{m,r,2^p} \times \A_{\T,2^p}$ are defined as follows. For $i,j\in[\![0,m{-}1]\!]$, $X,Y \in \{T,B\}$ and $d,e \in A_{2^p}$, there is a transition labeled by $(d,e)$ from the state $(i,X)$ to the state $(j,Y)$ if and only if
\[
	2^p i + e = md + j \quad \andrm \quad 
	Y = 	X_d.
\]
We denote by $\delta_\times$ the (partial) transition function of this product automaton. The initial state is $(0,T)$ and the only final state is $(r,T)$.

\begin{lemma}
\label{lem:transitionsProd}
For all $i,j\in[\![0,m{-}1]\!]$, $X ,Y \in \{T,B\}$ and $(u,v)\in (A_{2^p}\times A_{2^p})^*$, we have $\delta_\times((i,X),(u,v))=(j,Y)$ if and only if
\[
	2^{p\,|(u,v)|}\, i + \val_{2^p}(v) = m\, \val_{2^p}(u) + j 
	\quad \andrm \quad
	Y= X_{\val_{2^p}(u)}.
\]
\end{lemma}

\begin{proof}
It suffices to combine Lemmas~\ref{lem:transitionsTM} and~\ref{lem:transitionsAmb}.
\end{proof}

In Figure~\ref{fig:product}, we have depicted the automaton $\A_{6,2,4} \times \A_{\T,4}$, as well as the automata $\A_{6,2,4}$ and $\A_{\T,4}$, which we have placed in such a way that the labels of the product automaton can be easily deduced. Here and in the next figures, states are named $iX$ instead of $(i,X)$ for clarity. 

\begin{sidewaysfigure}
\begin{center}
\begin{minipage}{\linewidth} 
\begin{tikzpicture}[scale=0.8]
\tikzstyle{every node}=[shape=circle, fill=none, draw=black,
minimum size=30pt, inner sep=2pt]
\node(0) at (0,6.5) {$0$};
\node(1) at (3.5,6.5) {$1$};
\node(2) at (7,6.5) {$2$};
\node(3) at (10.5,6.5) {$3$};
\node(4) at (14,6.5) {$4$};
\node(5) at (17.5,6.5) {$5$};
\node(T) at (-5.5,0) {$T$};
\node(B) at (-5.5,-4.5) {$B$};
\node(0T) at (0,0) {$0T$};
\node(1T) at (3.5,0) {$1T$};
\node(2T) at (7,0) {$2T$};
\node(3T) at (10.5,0) {$3T$};
\node(4T) at (14,0) {$4T$};
\node(5T) at (17.5,0) {$5T$};
\node(0B) at (0,-4.5) {$0B$};
\node(1B) at (3.5,-4.5) {$1B$};
\node(2B) at (7,-4.5) {$2B$};
\node(3B) at (10.5,-4.5) {$3B$};
\node(4B) at (14,-4.5) {$4B$};
\node(5B) at (17.5,-4.5) {$5B$};
\tikzstyle{every node}=[shape=circle, fill=none, draw=black,
minimum size=25pt, inner sep=2pt]
\node at (7,0) {};
\node at (7,6.5) {};
\node at (-5.5,0) {};

\tikzstyle{etiquettedebut}=[very near start,rectangle,fill=black!20]
\tikzstyle{etiquettemilieu}=[midway,rectangle,fill=black!20]
\tikzstyle{every path}=[color=black, line width=0.5 pt]
\tikzstyle{every node}=[shape=circle, minimum size=5pt, inner sep=2pt]

\draw [->] (-1.5,6.5) to node {} (0); 
\draw [->] (0) to [loop above] node [left,rectangle,fill=black!20] {$(0,0)$} (0);
\draw [->] (1) to [loop below] node [left,rectangle,fill=black!20] {$(1,3)$} (1);
\draw [->] (2) to [loop left] node [left,rectangle,fill=black!20] {$(1,0)$} (2);
\draw [->] (3) to [loop below] node [left,rectangle,fill=black!20] {$(2,3)$} (3);
\draw [->] (4) to [loop above] node [left,rectangle,fill=black!20] {$(2,0)$} (4);
\draw [->] (5) to [loop above] node [left,rectangle,fill=black!20] {$(3,3)$} (5);
\draw [->] (0) to [bend left=15] node [sloped,etiquettemilieu] {$(0,1)$} (1);
\draw [->] (0) to [bend left=30] node [sloped,etiquettemilieu] {$(0,2)$} (2);
\draw [->] (0) to [bend left=45] node [sloped,etiquettedebut] {$(0,3)$} (3);
\draw [->] (1) to [bend left=30] node [sloped,pos=0.37,rectangle,fill=black!20] {$(0,0)$} (4);
\draw [->] (1) to [bend left=45] node [sloped,etiquettedebut] {$(0,1)$} (5);
\draw [->] (1) to [bend left=15] node [sloped,etiquettemilieu] {$(1,2)$} (0);
\draw [->] (2) to [bend left=15] node [sloped,etiquettemilieu] {$(1,1)$} (3);
\draw [->] (2) to [bend left=30] node [sloped,etiquettemilieu] {$(1,2)$} (4);
\draw [->] (2) to [bend left=45] node [sloped,etiquettemilieu] {$(1,3)$} (5);
\draw [->] (3) to [bend left=45] node [sloped,etiquettemilieu] {$(2,0)$} (0);
\draw [->] (3) to [bend left=40] node [sloped,pos=0.45,rectangle,fill=black!20] {$(2,1)$} (1);
\draw [->] (3) to [bend left=15] node [sloped,etiquettemilieu] {$(2,2)$} (2);
\draw [->] (4) to [bend left=15] node [sloped,etiquettemilieu] {$(2,1)$} (5);
\draw [->] (4) to [bend left=45] node [sloped,etiquettemilieu] {$(3,2)$} (0);
\draw [->] (4) to [bend left=40] node [sloped,pos=0.5,rectangle,fill=black!20] {$(3,3)$} (1);
\draw [->] (5) to [bend left=40] node [sloped,pos=0.45,rectangle,fill=black!20] {$(3,0)$} (2);
\draw [->] (5) to [bend left=35] node [sloped,etiquettemilieu] {$(3,1)$} (3);
\draw [->] (5) to [bend left=15] node [sloped,etiquettemilieu] {$(3,2)$} (4);

\draw [->] (-7,0) to node {} (T); 
\draw [->] (T) to [loop above] node [etiquettemilieu] {\begin{tabular}{c}
$(0,0),(0,1),(0,2),(0,3)$ \\
$(3,0),(3,1),(3,2),(3,3)$ \\
\end{tabular}} (T);
\draw [->] (B) to [loop below] node [etiquettemilieu] {\begin{tabular}{c}
$(0,0),(0,1),(0,2),(0,3)$ \\
$(3,0),(3,1),(3,2),(3,3)$ \\
\end{tabular}} (B);
\draw [->] (T) to [bend right=30] node [left=-0.3cm,etiquettemilieu] {\begin{tabular}{c}
$(1,0),(1,1),$\\$(1,2),(1,3)$ \\
$(2,0),(2,1),$\\$(2,2),(2,3)$ \\
\end{tabular}} (B);
\draw [->] (B) to [bend right=30] node [right=-0.3cm,etiquettemilieu] {\begin{tabular}{c}
$(1,0),(1,1),$\\$(1,2),(1,3)$ \\
$(2,0),(2,1),$\\$(2,2),(2,3)$ \\
\end{tabular}} (T);

\draw [->] (-1.5,0) to node {} (0T); 
\draw [->] (0T) to [loop above] node [] {} (0T);
\draw [->] (5T) to [loop above] node [] {} (5T);
\draw [->] (0B) to [loop below] node [] {} (0B);
\draw [->] (5B) to [loop below] node [] {} (5B);
\draw [->] (0T) to [] node [] {} (1T);
\draw [->] (0T) to [bend left=20] node [] {} (2T);
\draw [->] (0T) to [bend left=25] node [] {} (3T);
\draw [->] (1T) to [bend left=20] node [] {} (4T);
\draw [->] (1T) to [bend left=30] node [] {} (5T);
\draw [->] (1T) to [] node [] {} (0B); 
\draw [->] (1T) to [bend left=15] node [] {} (1B);
\draw [->] (2T) to [bend left=15] node [] {} (2B);
\draw [->] (2T) to [bend left=5] node [] {} (3B);
\draw [->] (2T) to [] node [] {} (4B); 
\draw [->] (2T) to [] node [] {} (5B);
\draw [->] (3T) to [] node [] {} (0B);
\draw [->] (3T) to [] node [] {} (1B);
\draw [->] (3T) to [bend left=5] node [] {} (2B);
\draw [->] (3T) to [bend left=15] node [] {} (3B); 
\draw [->] (4T) to [bend left=15] node [] {} (4B);
\draw [->] (4T) to [] node [] {} (5B);
\draw [->] (4T) to [bend right=35] node [] {} (0T);
\draw [->] (4T) to [bend right=25] node [] {} (1T); 
\draw [->] (5T) to [bend right=25] node [] {} (2T);
\draw [->] (5T) to [bend right=20] node [] {} (3T);
\draw [->] (5T) to [] node [] {} (4T); 

\draw [->] (0B) to [] node [] {} (1B); 
\draw [->] (0B) to [bend right=20] node [] {} (2B);
\draw [->] (0B) to [bend right=25] node [] {} (3B);
\draw [->] (1B) to [bend right=20] node [] {} (4B);
\draw [->] (1B) to [bend right=30] node [] {} (5B); 
\draw [->] (1B) to [] node [] {} (0T);
\draw [->] (1B) to [bend left=15] node [] {} (1T);
\draw [->] (2B) to [bend left=15] node [] {} (2T);
\draw [->] (2B) to [bend left=5] node [] {} (3T); 
\draw [->] (2B) to [] node [] {} (4T);
\draw [->] (2B) to [] node [] {} (5T);
\draw [->] (3B) to [] node [] {} (0T); 
\draw [->] (3B) to [] node [] {} (1T);
\draw [->] (3B) to [bend left=5] node [] {} (2T);
\draw [->] (3B) to [bend left=15] node [] {} (3T);
\draw [->] (4B) to [bend left=15] node [] {} (4T);
\draw [->] (4B) to [] node [] {} (5T); 
\draw [->] (4B) to [bend left=35] node [] {} (0B);
\draw [->] (4B) to [bend left=25] node [] {} (1B);
\draw [->] (5B) to [bend left=25] node [] {} (2B); 
\draw [->] (5B) to [bend left=20] node [] {} (3B);
\draw [->] (5B) to [] node [] {} (4B);
\end{tikzpicture}
\caption{The product automaton $\A_{6,2,4} \times \A_{\T,4}$}
\label{fig:product}
\end{minipage}
\end{center}
\end{sidewaysfigure}
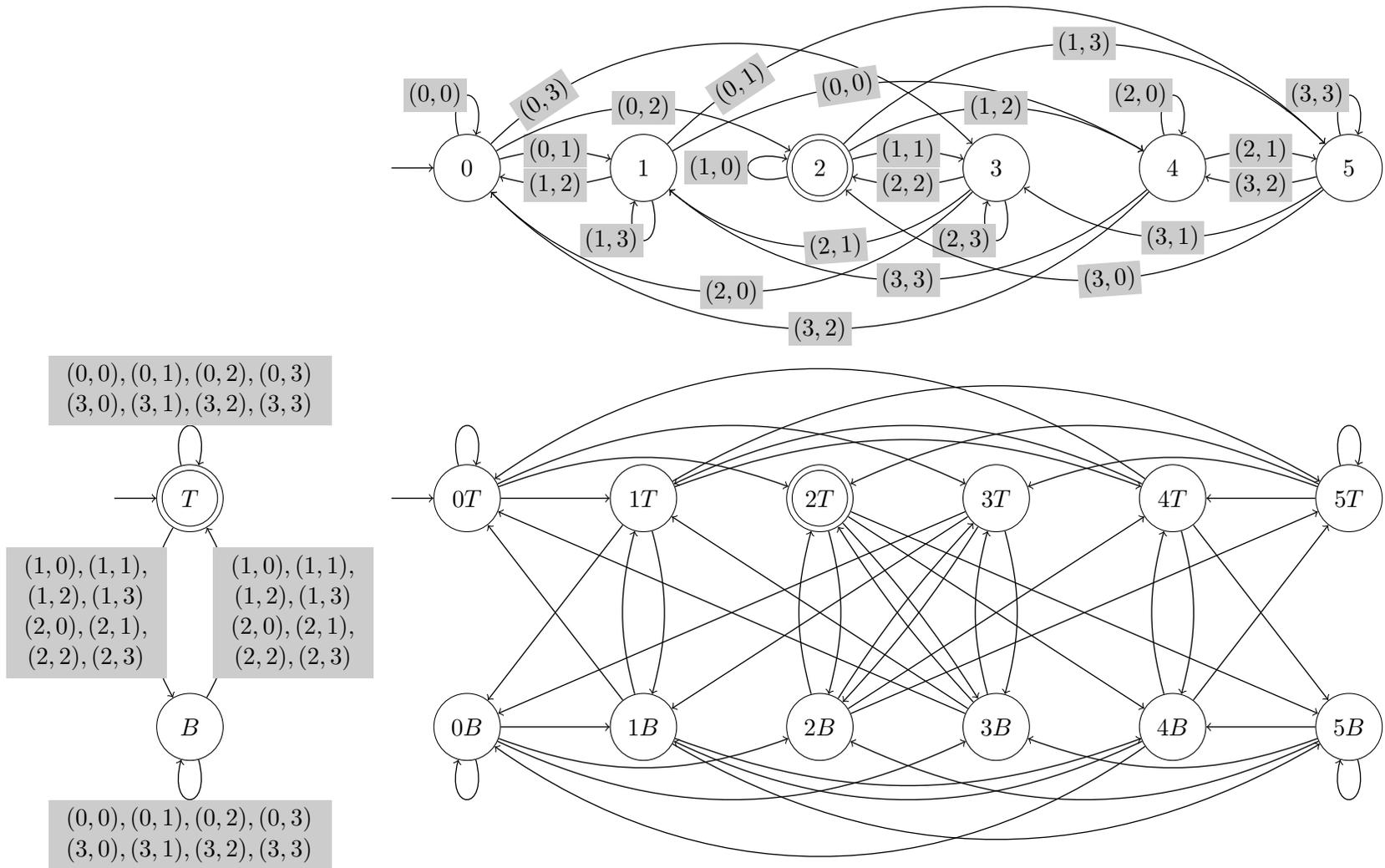

\subsection{The projection $\Pi \left( \A_{m,r,2^p} \times \A_{\T,2^p} \right)$ of the product automaton}

Now, we provide a DFA accepting the language $\val_{2^p}^{-1}(m\T+r)$. This automaton is denoted by $\Pi \left( \A_{m,r,2^p} \times \A_{\T,2^p} \right)$ and is defined from the automaton $\A_{m,r,2^p} \times \A_{\T,2^p}$ by only keeping the second component of the label of each transition. Formally, the states of $\Pi \left( \A_{m,r,2^p} \times \A_{\T,2^p} \right)$ are 
\[
	(0,T), \ldots , (m{-}1, T) \andrm (0,B) , \ldots , (m{-}1, B),
\]
the initial state is $(0,T)$, the only final state is $(r,T)$, and the transitions are defined as follows. For $i,j \in [\![0,m{-}1]\!]$, $X,Y \in \{T,B\}$ and $e\in A_{2^p}$, there is a transition labeled by $e$ from the state $(i,X)$ to the state $(j,Y)$ if and only if there exists $d \in A_{2^p}$ such that
\[
	2^p i + e = md + j \quad \andrm \quad 
	Y = 	X_d.
\]
We denote by $\delta_\times^\Pi$ the (partial) transition function of this product automaton. 

\begin{example}
The automata $\A_{6,2,4} \times \A_{\T,4}$ and $\Pi \left( \A_{6,2,4} \times \A_{\T,4} \right)$ are depicted in Figures~\ref{fig:product} and~\ref{fig:projected-aut} respectively. In Figure~\ref{fig:projected-aut}, all edges labeled by $0$ ($1,2$ and $3$ respectively) are represented in black (blue, red and green respectively).
\begin{figure}[htb]
\centering
\begin{tikzpicture}[scale=0.75]
\tikzstyle{every node}=[shape=circle, fill=none, draw=black,
minimum size=30pt, inner sep=2pt]
\node(0T) at (0,0) {$0T$};
\node(1T) at (3.5,0) {$1T$};
\node(2T) at (7,0) {$2T$};
\node(3T) at (10.5,0) {$3T$};
\node(4T) at (14,0) {$4T$};
\node(5T) at (17.5,0) {$5T$};
\node(0B) at (0,-4.5) {$0B$};
\node(1B) at (3.5,-4.5) {$1B$};
\node(2B) at (7,-4.5) {$2B$};
\node(3B) at (10.5,-4.5) {$3B$};
\node(4B) at (14,-4.5) {$4B$};
\node(5B) at (17.5,-4.5) {$5B$};
\tikzstyle{every node}=[shape=circle, fill=none, draw=black,
minimum size=25pt, inner sep=2pt]
\node at (7,0) {};

\tikzstyle{etiquettedebut}=[very near start,rectangle,fill=black!20]
\tikzstyle{etiquettemilieu}=[midway,rectangle,fill=black!20]
\tikzstyle{every path}=[color=black, line width=0.5 pt]
\tikzstyle{every node}=[shape=circle, minimum size=5pt, inner sep=2pt]
\draw [->] (-1.5,0) to node {} (0T); 
\draw [->] (0T) to [loop above] node [] {$0$} (0T);
\draw [green,->] (5T) to [loop above] node [] {} (5T);
\draw [->] (0B) to [loop below] node [] {} (0B);
\draw [green,->] (5B) to [loop below] node [] {} (5B);
\draw [blue,->] (0T) to [] node [above=-0.1] {$1$} (1T);
\draw [red,->] (0T) to [bend left=20] node [above=-0.1] {$2$} (2T);
\draw [green,->] (0T) to [bend left=25] node [above] {$3$} (3T);
\draw [->] (1T) to [bend left=20] node [] {} (4T);
\draw [blue,->] (1T) to [bend left=30] node [] {} (5T);
\draw [red,->] (1T) to [] node [] {} (0B);
\draw [green,->] (1T) to [bend left=15] node [] {} (1B);
\draw [->] (2T) to [bend left=15] node [] {} (2B);
\draw [blue,->] (2T) to [bend left=5] node [] {} (3B);
\draw [red,->] (2T) to [] node [] {} (4B);
\draw [green,->] (2T) to [] node [] {} (5B);
\draw [->] (3T) to [] node [] {} (0B);
\draw [blue,->] (3T) to [] node [] {} (1B);
\draw [red,->] (3T) to [bend left=5] node [] {} (2B);
\draw [green,->] (3T) to [bend left=15] node [] {} (3B);
\draw [->] (4T) to [bend left=15] node [] {} (4B);
\draw [blue,->] (4T) to [] node [] {} (5B);
\draw [red,->] (4T) to [bend right=35] node [] {} (0T);
\draw [green,->] (4T) to [bend right=25] node [] {} (1T);
\draw [->] (5T) to [bend right=25] node [] {} (2T);
\draw [blue,->] (5T) to [bend right=20] node [] {} (3T);
\draw [red,->] (5T) to [] node [] {} (4T);
\draw [blue,->] (0B) to [] node [] {} (1B);
\draw [red,->] (0B) to [bend right=20] node [] {} (2B);
\draw [green,->] (0B) to [bend right=25] node [] {} (3B);
\draw [->] (1B) to [bend right=20] node [] {} (4B);
\draw [blue,->] (1B) to [bend right=30] node [] {} (5B);
\draw [red,->] (1B) to [] node [] {} (0T);
\draw [green,->] (1B) to [bend left=15] node [] {} (1T);
\draw [->] (2B) to [bend left=15] node [] {} (2T);
\draw [blue,->] (2B) to [bend left=5] node [] {} (3T);
\draw [red,->] (2B) to [] node [] {} (4T);
\draw [green,->] (2B) to [] node [] {} (5T);
\draw [->] (3B) to [] node [] {} (0T);
\draw [blue,->] (3B) to [] node [] {} (1T);
\draw [red,->] (3B) to [bend left=5] node [] {} (2T);
\draw [green,->] (3B) to [bend left=15] node [] {} (3T);
\draw [->] (4B) to [bend left=15] node [] {} (4T);
\draw [blue,->] (4B) to [] node [] {} (5T);
\draw [red,->] (4B) to [bend left=35] node [] {} (0B);
\draw [green,->] (4B) to [bend left=25] node [] {} (1B);
\draw [->] (5B) to [bend left=25] node [] {} (2B);
\draw [blue,->] (5B) to [bend left=20] node [] {} (3B);
\draw [red,->] (5B) to [] node [] {} (4B);
\end{tikzpicture}
\caption{The projected automaton $\Pi \left( \A_{6,2,4} \times \A_{\T,4} \right)$.}
\label{fig:projected-aut}
\end{figure}
\end{example}

\begin{lemma}
\label{lem:transitionsProd-proj}
For all $i,j\in[\![0,m{-}1]\!]$, $X ,Y \in \{T,B\}$ and $v\in A_{2^p}^*$, we have
$\delta_\times^\Pi((i,X),v)=(j,Y)$ if and only if there exists $\ell\in\N$ 
such that
\[
	2^{p\,|v|}\, i + \val_{2^p}(v) = m\ell + j 
	\quad \andrm \quad
	Y= X_\ell.
\]
\end{lemma}

\begin{proof}
We have $\delta_\times^\Pi((i,X),v)=(j,Y)$ if and only if there exists some word $u$ over $A_{2^p}$ of the same length as $v$ such that $\delta_\times((i,X),(u,v))=(j,Y)$. Take $\ell=\val_{2^p}(u)$. The conclusion follows from Lemma~\ref{lem:transitionsProd} and a similar argument as in the proof of Lemma~\ref{lem:transitionsPiAmb}.
\end{proof}

\begin{remark}
\label{rem:2^pn}
Note that in the statement of Lemma~\ref{lem:transitionsProd-proj}, the integer $\ell$ is necessarily less than $2^{p|v|}$. This is due to the fact that if $v$ is the label of some path in the projected automaton $\Pi \left(\A_{m,r,2^p} \times \A_{\T,2^p} \right)$, there must be a word $u$ of the same length $\ell$ as $v$ such that the pair $(u,v)$ is the label of a path in the automaton $\A_{m,r,2^p} \times \A_{\T,2^p}$. This can be deduced directly from the computation: if $2^{p\,|v|}\, i + \val_{2^p}(v) = m\ell + j$ (with $i,j\in[\![0,m{-}1]\!]$) then 
\[
	\ell=\frac{2^{p\,|v|}\, i + \val_{2^p}(v)-j}{m}<\frac{2^{p\,|v|} (i +1)}{m}\le 2^{p\,|v|}.
\]
\end{remark}

\section{Properties of the intermediate automata}
\label{sec:properties-automata}

Now we prove some properties of the automata $\A_{\T,2^p}$, $\A_{m,r,b}$, $\Pi(\A_{m,r,b})$, $\A_{\T,2^p}\times\A_{m,r,b}$ and $\Pi(\A_{\T,2^p}\times\A_{m,r,b})$ that will be useful for our concerns.

\subsection{Properties of $\A_{\T,2^p}$}

\begin{lemma}
\label{lem:TM}
For all $X,Y \in \{T,B\}$ and $(u,v)\in (A_{2^p}\times A_{2^p})^*$, we have 
\[
	\delta_{\T,2^p}(X,(u,v))=Y\iff \delta_{\T,2^p}(\overline{X},(u,v))=\overline{Y}.
\]
\end{lemma}

\begin{proof}
This directly follows from Lemma~\ref{lem:transitionsTM}.
\end{proof}

\begin{proposition}
\label{prop:AT}
The automaton $\A_{\T,2^p}$ is complete, accessible, coaccessible and has disjoint states. In particular, it is the minimal automaton of $\val_{2^p}^{-1}(\T\times \N)$.
\end{proposition}

\begin{proof}
These properties are all straightforward verifications.
\end{proof}

\subsection{Properties of $\Pi(\A_{m,r,2^p})$ and $\A_{m,r,2^p}$}

In what follows, we let  $K=|\rep_{2^p}\big((k{-}1)2^z\big)|$, provided that $k>1$. Then we define a permutation  $\sigma$ of the integers in $[\![0,k{-}1]\!]$ by $\sigma(i)=- 2^{pK-z} i \bmod k$. Note that $\sigma$ permutes the integers $0,1,\ldots,k-1$ because $k$ is odd. Further, we define $w_i$ to be the unique word of length $K$ representing $\sigma(i)2^z$ in base $2^p$:
\[
	w_i=0^{K-|\rep_{2^p}(\sigma(i)2^z)|}\rep_{2^p}(\sigma(i)2^z)
\] 
for each $i\in [\![0,k{-}1]\!]$. Note that the words $w_i$ are well defined since, by the choice of $K$, we have $\sigma(i)2^z\le (k{-}1)2^z<2^{pK}$ for every $i\in[\![0,k{-}1]\!]$.

\begin{lemma}
\label{lem:k>1}
If $k>1$ then $pK\ge z$.
\end{lemma}

\begin{proof}
We have
\[
	K=\left\lfloor \log_{2^p}\big((k-1)2^z)\right\rfloor +1 
	=\left\lfloor \log_{2^p}(k-1)+\frac zp\right\rfloor +1
	\ge \left\lfloor \frac zp\right\rfloor+1
	\ge \left\lceil \frac zp\right\rceil.
\]
Thus $pK\geq p\big\lceil \frac zp\big\rceil \geq z$.
\end{proof}

\begin{lemma}
\label{lem:wi}
Suppose that $k>1$ and let $i\in[\![0,k{-}1]\!]$. Then the word $w_i$ leads from the state $i$ to the state $0$ in the automaton $\Pi(\A_{m,r,2^p})$. Otherwise stated, $\delta_{m,r,2^p}^\Pi(i,w_i)=0$.
\end{lemma}

\begin{proof}
The word $w_i$ has length $K$ and from Lemma~\ref{lem:k>1}, we know that $pK\ge z$. 
By Lemma~\ref{lem:transitionsPiAmb}, we have
\begin{align*}
 \delta_{m,r,2^p}^\Pi(i,w_i)=0
 		&\iff	2^{pK} i +\val_{2^p}(w_i)\equiv 0\pmod m \\
 		&\iff	2^{pK} i +\sigma(i)2^z\equiv 0\pmod{k2^z} \\
 		&\iff	2^{pK-z} i +\sigma(i)\equiv 0\pmod k.
	\end{align*}
The result follows from the definition of $\sigma$.
\end{proof}

\begin{lemma}
\label{lem:ij}
Suppose that $k>1$ and let $i,j,\in[\![0,k{-}1]\!]$. For any $\ell\in\N$, the word 
\[
	w_i(\rep_{2^p}(m))^\ell\rep_{2^p}(r)
\] 
is accepted from $j$ in the automaton $\Pi(\A_{m,r,2^p})$ if and only if $i=j$. 
\end{lemma}

\begin{proof}
Let $\ell\in\N$ and, for each $i\in[\![0,k{-}1]\!]$, let $y_i=w_i(\rep_{2^p}(m))^\ell\rep_{2^p}(r)$. Recall that we have set $R=|\rep_{2^p}(r)|$. Further, set $M=|\rep_{2^p}(m)|$. Then $|y_i|=K+\ell M+R$ and from Lemma~\ref{lem:k>1}, we know that $pK\ge z$. Therefore, we have
\begin{align*}
& \hphantom{\iff}\ \;  j 2^{p|y_i|} +\val_{2^p}(y_i) \equiv  r\pmod m \\
		&\iff 	j 2^{p(K+\ell M+R)} +\val_{2^p}(w_i)2^{p(\ell M+R)}+\sum_{s=0}^{\ell-1}m2^{p(sM+R)}+r\equiv r\pmod m \\
		&\iff 	j 2^{p(K+\ell M+R)} +\sigma(i)2^{z+p(\ell M+R)}\equiv 0\pmod{k2^z} \\	
		&\iff 	j 2^{p(K+\ell M+R)-z} + \sigma(i)2^{p(\ell M+R)}\equiv 0\pmod k \\
		&\iff 	j 2^{p(K+\ell M+R)-z} -i2^{p(K+\ell M+R)-z}\equiv 0\pmod k \\
		&\iff	j\equiv i \pmod k\\
		&\iff	j=i.
\end{align*}
The conclusion follows from Lemma~\ref{lem:transitionsPiAmb}.
\end{proof}

\begin{proposition}
\label{prop:PiAproperties}
The automaton $\Pi(\A_{m,r,2^p})$ is complete, accessible and coaccessible.
\end{proposition}

$\delta_{m,r,2^p}^\Pi(0,\rep_{2^p}(r))=r$ 

\begin{proof}
By Lemma~\ref{lem:transitionsPiAmb}, we have $\delta_{m,r,2^p}^\Pi(0,\rep_{2^p}(i))=i$ for all $i\in[\![0,m{-}1]\!]$. Thus, the automaton $\Pi(\A_{m,r,2^p})$ is accessible and in particular, we have $\delta_{m,r,2^p}^\Pi(0,\rep_{2^p}(r))=r$. By Lemma~\ref{lem:wi}, we have $\delta_{m,r,2^p}^\Pi(i,w_i)=0$ for all $i\in[\![0,m{-}1]\!]$, and hence, the concatenation $w_i\rep_{2^p}(r)$ leads from $i$ to $r$. This shows that $\Pi(\A_{m,r,2^p})$ is also coaccessible. In order to see that it is complete, observe that for every state $i\in[\![0,m{-}1]\!]$ and every digit $e\in A_{2^p}$, there is a transition labeled by $e$ from $i$ to the state $2^pi+e\bmod m$ (see Remark~\ref{rem:unicitelettrepremcomp}).
\end{proof}

In~\cite{Charlier&Cisternino&Massuir:2019}, we proved the coaccessibility for arbitrary integer bases $b$ by using a different method. In the present case of a base which is a power of two, the argument is more straightforward since we are able to explicitly provide a word that is accepted from any state $i$.

The automaton $\Pi(\A_{m,r,b})$ is not minimal in general: it is minimal if and only if $m$ and $b$ are coprime; see for example \cite{Alexeev:2004}. 


%
%
\begin{proposition} 
\label{prop:Ambproperties}
The automaton $\A_{m,r,2^p}$ is accessible, coaccessible and has disjoint states. 
\end{proposition}

\begin{proof} 
It directly follows from Proposition~\ref{prop:PiAproperties} that $\A_{m,r,2^p}$ is accessible and coaccessible. Now, let $i,j \in [\![0,m{-}1]\!]$ and let $(u,v) \in L_i\cap L_j$. By Lemma~\ref{lem:transitionsAmb}, we have
\[
	2^{p|(u,v)|} i + \val_b(v) = m\, \val_{2^p}(u)+r\quad  \andrm \quad 
	 2^{p|(u,v)|} j + \val_b(v) = m\,\val_{2^p}(u)+r,
\]
which implies that $i=j$. This proves that $\A_{m,r,2^p}$ has disjoint states.
\end{proof}

In a reduced DFA, there can be at most one non coaccessible state. Thus, we deduce from Proposition~\ref{prop:Ambproperties} that $\A_{m,r,2^p}$ is indeed the {\em trim minimal} automaton of the language $\val_{2^p}^{-1}\big(\{(n,mn+r)\colon n\in \N\}\big)$, that is the automaton obtained by removing the only non coaccessible state from its minimal automaton.

\subsection{Properties of $\A_{m,r,2^p} \times \A_{\T,2^p}$}

\begin{lemma}
\label{lem:TM-product}
For all $i,j\in[\![0,m{-}1]\!]$, $X,Y \in \{T,B\}$ and $(u,v)\in (A_{2^p}\times A_{2^p})^*$, we have 
\[
	\delta_{\T,2^p}((i,X),(u,v))=(j,Y)\iff \delta_{\T,2^p}((i,\overline{X}),(u,v))=(j,\overline{Y}).
\]
\end{lemma}

\begin{proof}
This directly follows from Lemma~\ref{lem:TM}.
\end{proof}

\begin{lemma}
\label{lem:1-m}
Let $i\in[\![0,m{-}1]\!]$ and $X\in\{T,B\}$. The word $\rep_{2^p}(1,m)$ is the label of a path from the state $(0,X)$ to the state $(0,\overline{X})$ in $\A_{m,r,2^p} \times \A_{\T,2^p}$. 
\end{lemma}

\begin{proof}
This directly follows from Lemma~\ref{lem:transitionsProd}. 
\end{proof}



\begin{proposition}
\label{prop:prodaccessible}
The automaton $\A_{m,r,2^p} \times \A_{\T,2^p}$ accepts $\val_{2^p}^{-1}(\{(t,mt+r)\colon t\in\T\})$, is accessible, coaccessible and has disjoint states. 
\end{proposition}

\begin{proof}
By construction of the product automaton and since
\[
	\{(n,mn+r)\colon n\in \N\}
	\cap \big(\T\times \N\big)
	=\{(t,mt+r)\colon t\in\T\},
\]
we get that the product automaton $\A_{m,r,2^p} \times \A_{\T,2^p}$ accepts the language
\[
	\val_{2^p}^{-1}(\{(t,mt+r)\colon t\in\T\}).
\]
By Lemma~\ref{lem:transitionsProd}, we can check that for every $i\in[\![0,m{-}1]\!]$, the states $(i,T)$ and $(i,B)$ are  accessible thanks to the word $\rep_{2^p}(0,i)$ and $\rep_{2^p}(1,m+i)$ respectively. Hence, $\A_{m,r,2^p} \times \A_{\T,2^p}$ is accessible. 

To show the coaccessibility, we now fix some $i\in[\![0,m{-}1]\!]$ and $X \in \{T,B \}$. By Lemma~\ref{lem:wi}, we know that there exists a word $w_i$ that leads from the state $i$ to the state $0$ in the automaton $\Pi(\A_{m,r,2^p})$. Thus, there exists a word $u$ of the same length as $w_i$ such that the word $(u,w_i)$ leads from $i$ to $0$ in $\A_{m,r,2^p}$. Now, by reading $(u,w_i)$ from $(i,X)$ in $\A_{m,r,2^p} \times \A_{\T,2^p}$, we reach either the state $(0,T)$ or the state $(0,B)$. If we reach $(0,T)$, then the concatenation $(u,w_i)\rep_b(0,r)$ leads from the  state $(i,X)$ to $(r,T)$ in $\A_{m,r,2^p} \times \A_{\T,2^p}$. If we reach $(0,B)$ instead, then we may apply Lemma~\ref{lem:1-m} in order to obtain that the concatenation $(u,w_i)\rep_b(1,m)\rep_b(0,r)$ leads from $(i,X)$ to $(r,T)$ in $\A_{m,r,2^p} \times \A_{\T,2^p}$. This proves that $\A_{m,r,2^p} \times \A_{\T,2^p}$ is coaccessible. 

The fact that $\A_{m,r,2^p} \times \A_{\T,2^p}$ has disjoint states follows from Propositions~\ref{prop:AT} and~\ref{prop:Ambproperties}.
\end{proof}

\subsection{Properties of $\Pi(\A_{m,r,2^p} \times \A_{\T,2^p})$}

\begin{lemma}
\label{lem:projdisj-0}
For all $i,j \in [\![0,m{-}1]\!]$, $X,Y\in\{T,B\}$ and $v\in A_{2^p}^*$, we have
\[
	\delta_\times^\Pi\big((i,X),v\big)=(j,Y)\implies
	\delta_{m,r,2^p}^\Pi(i,v)=j.
\]
\end{lemma}

\begin{proof}
This is a direct verification.
\end{proof}

\begin{lemma}
\label{lem:projdisj}
For every $i \in [\![0,m{-}1]\!]$, the states $(i,T)$ and $(i,B)$ are disjoint in the projected automaton $\Pi \left( \A_{m,r,2^p} \times \A_{\T,2^p} \right)$.
\end{lemma}

\begin{proof}
Proceed by contradiction and suppose that a word $v$ over $A_{2^p}$ is accepted from both $(i,T)$ and $(i,B)$ in $\Pi \left( \A_{m,r,2^p} \times \A_{\T,2^p} \right)$ for some $i \in [\![0,m{-}1]\!]$. Then there exists words $u$ and $u'$ over $A_{2^p}$ of length $|v|$ such that,  in $\A_{m,r,2^p} \times \A_{\T,2^p}$, the words $(u,v)$ and $(u',v)$ are accepted from $(i,T)$ and $(i,B)$ respectively. But from Remark~\ref{rem:unicitepremcomp}, we must have $u=u'$. Hence the word $(u,v)$ is accepted from both $(i,T)$ and $(i,B)$ in $\A_{m,r,2^p} \times \A_{\T,2^p}$, contradicting that this automaton has disjoint states (see Proposition~\ref{prop:prodaccessible}).
\end{proof}

\begin{proposition}
\label{prop:m-odd}
The automaton $\Pi \left(\A_{m,r,2^p}\times \A_{\T,2^p} \right)$ accepts $\val_{2^p}^{-1}(m\T+r)$, is deterministic, complete, accessible and coaccessible. 
\end{proposition}

\begin{proof}
By construction, $\Pi \left( \A_{m,r,2^p} \times \A_{\T,2^p} \right)$ accepts $\val_{2^p}^{-1}(m\T+r)$; see Section~\ref{sec:method}. The fact that this automaton is deterministic and complete follows from Remark~\ref{rem:unicitelettrepremcomp}. It is accessible and coaccessible because so is $\A_{m,r,2^p} \times \A_{\T,2^p}$.  
\end{proof}

\section{Minimization of $\Pi \left( \A_{m,r,2^p} \times \A_{\T,2^p} \right)$}
\label{sec:minimization}

We start by defining some classes of states of $\Pi \left( \A_{m,r,2^p} \times \A_{\T,2^p} \right)$. Our aim is twofold. First, we will prove that those classes consist in {\em indistinguishable} states, i.e.\ accepting the same language. Second, we will show that states belonging to different classes are {\em distinguishable}, i.e.\ accept different languages. Otherwise stated, these classes correspond to the left quotients $w^{{-}1}L$ where $w$ is any finite word over the alphabet $A_{2^p}$ and $L=\val^{-1}_{2^p}(m\T+r)$.

\subsection{Definition of the classes}
\label{sec:def-classes}

Recall that $R=|\rep_{2^p}(r)|$. We define $N=\max\{\big\lceil \frac{z}{p}\big\rceil,R\}$. The classes we are going to define are closely related to the base $2^p$-expansion of the remainder $r$ with some additional leading zeroes. More precisely, we have to consider the word $0^{N-R}\rep_{2^p}(r)$, which is the unique word over the alphabet $A_{2^p}$ with length $N$ and $2^p$-value $r$. This word is equal to the $2^p$-expansion $\rep_{2^p}(r)$ if and only if $N=R$, i.e.\ $\big\lceil \frac{z}{p}\big\rceil\le R$.


\begin{definition}
For $\alpha\in[\![0,N]\!]$, we define
\[
	C_\alpha' = \begin{cases}
					\big\{\big(\llfloor\frac{r}{2^{p\alpha}}\rrfloor+\ell \frac{m}{2^{p\alpha}}, T_\ell\big) \colon 0 \le \ell \le 2^{p\alpha}{-}1\big\}  				
					& \text{if } \alpha \le \frac{z}{p}\\
					\big\{\big(\llfloor\frac{r}{2^{p\alpha}}\rrfloor+\ell k, T_\ell\big) \colon 0 \le \ell \le 2^z{-}1\big\}  
					& \text{if } \alpha \ge \frac{z}{p}.
				\end{cases}
\]
\end{definition}

Note that in the case where $z$ is divisible by $p$, the two cases of the definition coincide for the value $\alpha=\frac{z}{p}$.

Let us comment the previous definition, which may seem quite technical at first. The first elements of the sets $C_\alpha'$ are the integer part of the remainder $r$ divided by increasing powers of the base $2^p$, i.e.\ $r$ divided by $2^{p\alpha}$ for the set indexed by $\alpha$. The further elements of a set $C_\alpha'$ are obtained by adding to the first element $\llfloor\frac{r}{2^{p\alpha}}\rrfloor$ integer multiples of $\frac{m}{2^{p\alpha}}$ so that the greatest element so-obtained is still less than $m$, provided that $m$ is divisible by this power $2^{p\alpha}$. When $m$ is no longer divisible by $2^{p\alpha}$, i.e.\ when $\alpha > \frac{z}{p}$, then we add integer multiples of $k$, which is the odd part of $m$. In particular, if $m$ is odd, i.e.\ if $z=0$, then all the sets $C_\alpha'$ are reduced to a single state: $C_\alpha'=\{(\llfloor\frac{r}{2^{p\alpha}}\rrfloor,T)\}$. Finally, note that since $R=|\rep_{2^p}(r)|$ and $N\ge R$, we have $\llfloor\frac{r}{2^{pN}}\rrfloor=0$ and $C_N'=\{(k\ell,T_\ell)\colon 0 \le \ell \le 2^z{-}1\}$.

We will see in Lemma~\ref{lem:suffixes} that the states in $C_\alpha'$ are exactly those from which there is a path labeled by the suffix of length $\alpha$ of $0^{N-R}\rep_{2^p}(r)$ to the state $(r,T)$. 

\begin{example}
\label{ex:c-alpha'}
Let $m=24$ and $p=2$. We have $k=3$ and $z=3$. Let us consider the extremal possible values of the remainder $r$. For $r=23$, we have $\rep_4(23)=113$, $R=3$ and $N=\max\{\lceil\frac{3}{2}\rceil,3\}=3$. Thus, the sets defined above are
\begin{align*}
	C_0' & = \{(23,T)\} \\
	C_1' & = \{(5,T),(11,B),(17,B),(23,T)\} \\ 
	C_2' & = \{(1,T),(4,B),(7,B),(10,T),(13,B),(16,T),(19,T),(22,B)\} \\ 
	C_3' & = \{(0,T),(3,B),(6,B),(9,T),(12,B),(15,T),(18,T),(21,B)\}. 
\end{align*}
For instance, we have the following paths, which are labeled by $0^{N-R}\rep_4(23)=113$:
\[
	(0,T)\overset{1}{\longrightarrow} (1,T)\overset{1}{\longrightarrow} (5,T)\overset{3}{\longrightarrow} (23,T)
\]
and
\[
	(3,B)\overset{1}{\longrightarrow} (13,B)\overset{1}{\longrightarrow} (5,T)\overset{3}{\longrightarrow} (23,T).
\]
For $r=0$, we have $R=0$ and $N=\max\{\lceil\frac{3}{2}\rceil,0\}=2$. In this case, the sets $C_\alpha'$ are
\begin{align*}
	C_0' & = \{(0,T)\} \\
	C_1' & = \{(0,T),(6,B),(12,B),(18,T)\} \\ 
	C_2' & = \{(0,T),(3,B),(6,B),(9,T),(12,B),(15,T),(18,T),(21,B)\}.
\end{align*}
For instance, we have the following paths, which are labeled by $0^{N-R}\rep_4(0)=00$:
\[
	(0,T)\overset{0}{\longrightarrow} (0,T)\overset{0}{\longrightarrow} (0,T)
\]
and
\[
	(3,B)\overset{0}{\longrightarrow} (12,B)\overset{0}{\longrightarrow} (0,T).
\]
\end{example}

The sets $C_\alpha'$ are not necessarily disjoint as Example~\ref{ex:c-alpha'} shows.
In order to obtain the desired classes of states of the automaton $\Pi \left( \A_{m,r,2^p} \times \A_{\T,2^p} \right)$, we consider the following definition. 

\begin{definition}
\label{def:C}
For $\alpha\in[\![0,N]\!]$, we define
\[
	C_\alpha=C_\alpha'\setminus \bigcup_{\beta=0}^{\alpha-1} C_\beta'.
\]
\end{definition}

Let us define a second type of classes. The idea behind this definition is that these classes are "too far" from the remainder $r$ with respect to the division by consecutive powers of the base $2^p$, in the sense that these states do not accept any suffix of  $0^{N-R}\rep_{2^p}(r)$.

\begin{definition}
\label{def:D}
For $(j,X)\in\big([\![0,k-1]\!] \times \{T,B\}\big)\setminus \{(0,T)\}$, we define
\[
	D_{(j,X)}' = \{(j+\ell k,X_\ell) \colon 0 \le \ell \le 2^z {-}1\}
\]
and 
\[
	D_{(j,X)} = D_{(j,X)}'\setminus \bigcup_{\alpha=0}^N C_\alpha.
\]
\end{definition}

As we already observed,  all states of the form $(\ell k,T_\ell)$ appear in the set $C_N'$ and thus, also in the union of the sets $C_\alpha$. This is the reason why the sets $D_{(0,T)}'$ and $D_{(0,T)}$ are not defined, i.e.\ $(j,X)\ne (0,T)$ in the previous definition. 

We will refer to the sets of states $C_\alpha$ and $D_{(j,X)}$ as {\em classes} of states. Let us make some preliminary observations concerning the previous definitions.

The classes $C_\alpha$ and $D_{(j,X)}$ are pairwise disjoint: the intersection of any two such classes is empty.
Moreover, the nonempty classes $C_\alpha$ and $D_{(j,X)}$ form a partition of the whole set of states $[\![0,m-1]\!]\times\{T,B\}$ of $\Pi \left( \A_{m,r,2^p} \times \A_{\T,2^p} \right)$. Note that if $m$ is odd, i.e.\ if $z =0$, then the sets $D_{(j,X)}'$ are reduced to a single state. If $m$ is a power of $2$, i.e.\  if $k=1$, then there is no set of the form $D_{(j,X)}'$ and $D_{(j,X)}$ with $j\ge 1$.

\begin{example}
Let us resume Example~\ref{ex:c-alpha'}.
For $r=23$, the classes defined above are
\begin{align*}
	C_0 & = \{(23,T)\} \\
	C_1 & = \{(5,T),(11,B),(17,B)\} \\ 
	C_2 & = \{(1,T),(4,B),(7,B),(10,T),(13,B),(16,T),(19,T),(22,B)\} \\ 
	C_3 & = \{(0,T),(3,B),(6,B),(9,T),(12,B),(15,T),(18,T),(21,B)\} \\ 
	D_{(1,T)} & =\emptyset  \\ 
	D_{(2,T)} & =\{(2,T),(5,B),(8,B),(11,T),(14,B),(17,T),(20,T),(23,B)\} \\
	D_{(0,B)} & =\{(0,B),(3,T),(6,T),(9,B),(12,T),(15,B),(18,B),(21,T)\} \\
	D_{(1,B)} & =\{(1,B),(4,T),(7,T),(10,B),(13,T),(16,B),(19,B),(22,T)\} \\
	D_{(2,B)} & =\{(2,B),(8,T),(14,T),(20,B)\}.
\end{align*}
Note that $2k+\lceil\frac{z}{p}\rceil=2\cdot 3+\lceil\frac{3}{2}\rceil=8$ of them are nonempty. In Figure~\ref{fig:projected-aut24-23}, the automaton $\Pi \left( \A_{24,23,4} \times \A_{\T,4} \right)$ is represented without the transitions and the states of are colored with respect to these classes.
\begin{figure}[!ht]
\centering
\begin{tikzpicture}
[scale=0.225]
\tikzstyle{every node}=[shape=circle, fill=none, draw=black,
minimum size=10pt, inner sep=2pt]
\node[fill=violet](0T) at (0,0) {};
\node[fill=cyan](1T) at (2,0) {};
\node[fill=gray](2T) at (4,0) {};
\node[fill=yellow](3T) at (6,0) {};
\node[fill=magenta](4T) at (8,0) {};
\node[fill=vert](5T) at (10,0) {};
\node[fill=yellow](6T) at (12,0) {};
\node[fill=magenta](7T) at (14,0) {};
\node[fill=orange](8T) at (16,0) {};
\node[fill=violet](9T) at (18,0) {};
\node[fill=cyan](10T) at (20,0) {};
\node[fill=gray](11T) at (22,0) {};
\node[fill=yellow](12T) at (24,0) {};
\node[fill=magenta](13T) at (26,0) {};
\node[fill=orange](14T) at (28,0) {};
\node[fill=violet](15T) at (30,0) {};
\node[fill=cyan](16T) at (32,0) {};
\node[fill=gray](17T) at (34,0) {};
\node[fill=violet](18T) at (36,0) {};
\node[fill=cyan](19T) at (38,0) {};
\node[fill=gray](20T) at (40,0) {};
\node[fill=yellow](21T) at (42,0) {};
\node[fill=magenta](22T) at (44,0) {};
\node(23T) at (46,0) {};

\node[fill=yellow](0B) at (0,-4.5) {};
\node[fill=magenta](1B) at (2,-4.5) {};
\node[fill=orange](2B) at (4,-4.5) {};
\node[fill=violet](3B) at (6,-4.5) {};
\node[fill=cyan](4B) at (8,-4.5) {};
\node[fill=gray](5B) at (10,-4.5) {};
\node[fill=violet](6B) at (12,-4.5) {};
\node[fill=cyan](7B) at (14,-4.5) {};
\node[fill=gray](8B) at (16,-4.5) {};
\node[fill=yellow](9B) at (18,-4.5) {};
\node[fill=magenta](10B) at (20,-4.5) {};
\node[fill=vert](11B) at (22,-4.5) {};
\node[fill=violet](12B) at (24,-4.5) {};
\node[fill=cyan](13B) at (26,-4.5) {};
\node[fill=gray](14B) at (28,-4.5) {};
\node[fill=yellow](15B) at (30,-4.5) {};
\node[fill=magenta](16B) at (32,-4.5) {};
\node[fill=vert](17B) at (34,-4.5) {};
\node[fill=yellow](18B) at (36,-4.5) {};
\node[fill=magenta](19B) at (38,-4.5) {};
\node[fill=orange](20B) at (40,-4.5) {};
\node[fill=violet](21B) at (42,-4.5) {};
\node[fill=cyan](22B) at (44,-4.5) {};
\node[fill=gray](23B) at (46,-4.5) {};

\tikzstyle{every node}=[shape=circle, fill=none, draw=black,
minimum size=7pt, inner sep=2pt]
\node at (46,0) {};
\end{tikzpicture}
\caption{The classes of the projected automaton $\Pi \left( \A_{24,23,4} \times \A_{\T,4} \right)$.}
\label{fig:projected-aut24-23}
\end{figure}
If now we consider $r=0$, these classes are
\begin{align*}
	C_0 & = \{(0,T)\} \\
	C_1 & = \{(6,B),(12,B),(18,T)\} \\
	C_2 & = \{(3,B),(9,T),(15,T),(21,B)\} \\
	D_{(1,T)} & =\{(1,T),(4,B),(7,B),(10,T),(13,B),(16,T),(19,T),(22,B)\}  \\
	D_{(2,T)} & =\{(2,T),(5,B),(8,B),(11,T),(14,B),(17,T),(20,T),(23,B)\} \\
	D_{(0,B)} & =\{(0,B),(3,T),(6,T),(9,B),(12,T),(15,B),(18,B),(21,T)\} \\
	D_{(1,B)} & =\{(1,B),(4,T),(7,T),(10,B),(13,T),(16,B),(19,B),(22,T)\} \\
	D_{(2,B)} & =\{(2,B),(5,T),(8,T),(11,B),(14,T),(17,B),(20,B),(23,T)\}.
\end{align*}
In this case, they are all nonempty and there are $8$ of them. In Figure~\ref{fig:projected-aut24}, the states of the automaton $\Pi \left( \A_{24,0,4} \times \A_{\T,4} \right)$ are colored with respect to these classes.
\begin{figure}[!ht]
\centering
\begin{tikzpicture}
[scale=0.225]
\tikzstyle{every node}=[shape=circle, fill=none, draw=black,
minimum size=10pt, inner sep=2pt]
\node(0T) at (0,0) {};
\node[fill=cyan](1T) at (2,0) {};
\node[fill=gray](2T) at (4,0) {};
\node[fill=yellow](3T) at (6,0) {};
\node[fill=magenta](4T) at (8,0) {};
\node[fill=orange](5T) at (10,0) {};
\node[fill=yellow](6T) at (12,0) {};
\node[fill=magenta](7T) at (14,0) {};
\node[fill=orange](8T) at (16,0) {};
\node[fill=violet](9T) at (18,0) {};
\node[fill=cyan](10T) at (20,0) {};
\node[fill=gray](11T) at (22,0) {};
\node[fill=yellow](12T) at (24,0) {};
\node[fill=magenta](13T) at (26,0) {};
\node[fill=orange](14T) at (28,0) {};
\node[fill=violet](15T) at (30,0) {};
\node[fill=cyan](16T) at (32,0) {};
\node[fill=gray](17T) at (34,0) {};
\node[fill=vert](18T) at (36,0) {};
\node[fill=cyan](19T) at (38,0) {};
\node[fill=gray](20T) at (40,0) {};
\node[fill=yellow](21T) at (42,0) {};
\node[fill=magenta](22T) at (44,0) {};
\node[fill=orange](23T) at (46,0) {};

\node[fill=yellow](0B) at (0,-4.5) {};
\node[fill=magenta](1B) at (2,-4.5) {};
\node[fill=orange](2B) at (4,-4.5) {};
\node[fill=violet](3B) at (6,-4.5) {};
\node[fill=cyan](4B) at (8,-4.5) {};
\node[fill=gray](5B) at (10,-4.5) {};
\node[fill=vert](6B) at (12,-4.5) {};
\node[fill=cyan](7B) at (14,-4.5) {};
\node[fill=gray](8B) at (16,-4.5) {};
\node[fill=yellow](9B) at (18,-4.5) {};
\node[fill=magenta](10B) at (20,-4.5) {};
\node[fill=orange](11B) at (22,-4.5) {};
\node[fill=vert](12B) at (24,-4.5) {};
\node[fill=cyan](13B) at (26,-4.5) {};
\node[fill=gray](14B) at (28,-4.5) {};
\node[fill=yellow](15B) at (30,-4.5) {};
\node[fill=magenta](16B) at (32,-4.5) {};
\node[fill=orange](17B) at (34,-4.5) {};
\node[fill=yellow](18B) at (36,-4.5) {};
\node[fill=magenta](19B) at (38,-4.5) {};
\node[fill=orange](20B) at (40,-4.5) {};
\node[fill=violet](21B) at (42,-4.5) {};
\node[fill=cyan](22B) at (44,-4.5) {};
\node[fill=gray](23B) at (46,-4.5) {};

\tikzstyle{every node}=[shape=circle, fill=none, draw=black,
minimum size=7pt, inner sep=2pt]
\node at (0,0) {};
\end{tikzpicture}
\caption{The classes of the projected automaton $\Pi \left( \A_{24,0,4} \times \A_{\T,4} \right)$.}
\label{fig:projected-aut24}
\end{figure}
\end{example}

Our aim is to prove that the nonempty classes defined above correspond exactly to the left quotients of the language $\val^{-1}_{2^p}(m\T+r)$.

\subsection{Research of the empty classes}
\label{sec:counting}

\begin{proposition}
\label{prop:nonempty-C}
For all $\alpha\in[\![0,N]\!]$, the classes $C_\alpha$ are nonempty. 
\end{proposition}

\begin{proof}
The sequence $\big(\llfloor\frac{r}{2^{p\alpha}}\rrfloor\big)_{\alpha\in[\![0,N]\!]}$ is (strictly) decreasing for $\alpha\in[\![0,R]\!]$ and is equal to~$0$ for $\alpha\in[\![R,N]\!]$. In particular, note that $\alpha=R$ is the first value for which $\llfloor\frac{r}{2^{p\alpha}}\rrfloor=0$. Therefore $(\llfloor\frac{r}{2^{p\alpha}}\rrfloor,T_0)=(\llfloor\frac{r}{2^{p\alpha}}\rrfloor,T)\in C_\alpha$ for every $\alpha\in[\![0,R]\!]$; see Figure~\ref{fig:nonempty-classes}. 
\begin{figure}[htb]
\[
\begin{array}{c|ccc}
C_0' & (r,T) \\
C_1' & (\llfloor\frac{r}{2^p}\rrfloor,T) \\
\vdots & \vdots \\
C_R' & (0,T) \\
\hline
C_{R+1}' & (0,T) & (\frac{m}{2^{p(R+1)}},B) \\
C_{R+2}' & (0,T) & (\frac{m}{2^{p(R+2)}},B)\\
\vdots & \vdots \\
C_{N}'	& (0,T) & (\frac{m}{2^{pN}},B) 
\end{array}
\]
\caption{The elements of the sets $C_\alpha'$ up to the first belonging to the classes $C_\alpha$.}
\label{fig:nonempty-classes}
\end{figure} 
If $N=R$ then we are done (in this case, the part of Figure~\ref{fig:nonempty-classes} below the line is empty).

Now suppose that $N=\lceil\frac{z}{p}\rceil>R$. Then $(\llfloor\frac{r}{2^{p\alpha}}\rrfloor,T_0)=(0,T)\in C_\alpha'\setminus C_\alpha$ for every $\alpha\in[\![R+1,N]\!]$. We claim that, for each such $\alpha$, the "next" element of $C_\alpha'$ (i.e.\ the element corresponding to $\ell=1$) indeed belongs to $C_\alpha$. Let us fix some $\alpha\in[\![R+1,N]\!]$.

First, suppose that $p\alpha\le z$. Then we have to prove that $(\frac{m}{2^{p\alpha}},T_1)=(\frac{m}{2^{p\alpha}},B)\notin C_\beta'$ for $\beta<\alpha$.  If $\alpha>\beta$ then $\frac{m}{2^{p\alpha}}<\frac{m}{2^{p\beta}}\le \llfloor\frac{r}{2^{p\beta}}\rrfloor+\frac{m}{2^{p\beta}}$. This shows that if the state $(\frac{m}{2^{p\alpha}},B)$ belongs to some $C_\beta'$ with $\beta<\alpha$, then its first component has to be $\llfloor\frac{r}{2^{p\beta}}\rrfloor$. But since it has second component $T_1=B$ and since $T_0=T$, the state $(\frac{m}{2^{p\alpha}},B)$ cannot be the first state of any set $C_\beta'$. 

Second, suppose that $p\alpha> z$. In this case, we have to prove that $(k,B)\notin C_\beta'$ for $\beta<\alpha$. Similarly to what precedes, if $(k,B)$ belongs to some set $C_\beta'$, then we must have $\llfloor \frac{r}{2^{p\beta}}\rrfloor=0$ and $\beta\ge\frac{z}{p}$. But since $\frac{z}{p}<\alpha\le N=\lceil\frac{z}{p}\rceil$, we get $\alpha=\llceil \frac{z}{p}\rrceil$. Therefore, any $\beta<\alpha$ is such that $p\beta< z$. 
\end{proof}

\begin{lemma}
\label{lem:number-classes}
We have $|\rep_{2^p}(m{-}1)|- \llceil\frac{z}{p} \rrceil\in[\![0,k{-}1]\!]$.
\end{lemma}

\begin{proof}
Observe that
\begin{equation}
\label{eq:repm}
	|\rep_{2^p}(m)|
	=\llfloor\log_{2^p}(m)\rrfloor+1
	=\left\lfloor\log_{2^p}(k)+\frac{z}{p}\right\rfloor+1.
\end{equation}
If $m$ is a power of $2^p$, i.e.\ if $k=1$ and $p$ divides $z$, then $|\rep_{2^p}(m{-}1)|=|\rep_{2^p}(m)|-1=\frac{z}{p}$ and the result is clear.

Now, suppose that $m$ is not a power of the base $2^p$. Then $|\rep_{2^p}(m{-}1)|=|\rep_{2^p}(m)|$. From \eqref{eq:repm} we get that $|\rep_{2^p}(m)|\ge \llceil\frac{z}{p}\rrceil$. Let us show that $|\rep_{2^p}(m)|-\llceil\frac{z}{p}\rrceil\le k-1$. If $k=1$ then $p$ does not divide $z$ (for otherwise $m$ would be a power of $2^p$) and we get from \eqref{eq:repm} that $|\rep_{2^p}(m)|= \llfloor\frac{z}{p}\rrfloor+1=\llceil\frac{z}{p} \rrceil$. 
If $k=3$ and $p=1$, then we obtain $|\rep_{2^p}(m)|=\llfloor\log_{2}(3)+z\rrfloor+1 =2+z=k-1+\llceil \frac{z}{p}\rrceil$. In all other cases, that is if $k\ge 5$ or ($k=3$ and $p\ge 2$), we can check that $\log_{2^p}(k)<k-2$. Therefore we have $|\rep_{2^p}(m)|\le \left\lfloor k-2+\frac{z}{p}\right\rfloor +1=k-1+\llfloor \frac{z}{p}\rrfloor$.
\end{proof}

\begin{proposition}
\label{prop:empty-D}
The empty classes $D_{(j,X)}$ are exactly those of the form $D_{(\lfloor\frac{r}{2^{p\alpha}}\rfloor,T)}$ with $p\alpha\ge z$.
\end{proposition}

\begin{proof}
The $k$ classes $D_{(0,B)},\ldots,D_{(k-1,B)}$ are all nonempty since for any $j\in[\![0,k-1]\!]$, the state $(j,B)$ does not belong to any $C_\alpha$. 

Now, let $j\in[\![1,k-1]\!]$. We have to show that all classes of the form $D_{(\lfloor\frac{r}{2^{p\alpha}}\rfloor,T)}$ with $p\alpha\ge z$ are empty (observe that if $p\alpha\ge z$ then $\llfloor\frac{r}{2^{p\alpha}}\rrfloor\le k-1$), and that the other classes $D_{(j,T)}$ are nonempty. 

If $(j,T)\notin \cup_{\alpha=0}^N C_\alpha'$, then the class $D_{(j,T)}$ is nonempty since it contains $(j,T)$. In this case, $j\ne \lfloor\frac{r}{2^{p\alpha}}\rfloor$ for any $\alpha\in[\![0,N]\!]$. Now, suppose that there exists some $\alpha\in[\![0,N]\!]$ such that $(j,T)\in C_\alpha'$. Since $j<k$, this $\alpha$ is unique and $j=\lfloor\frac{r}{2^{p\alpha}}\rfloor$. We have to show that $D_{(j,T)}$ is empty if and only if $p\alpha\ge z$. 

Clearly, $p\alpha\ge z$ implies that $D_{(j,T)}'=\{(j+\ell k,T_\ell)\colon\ell\in[\![0,2^z-1]\!]\}=C_\alpha'$, and hence that $D_{(j,T)}$ is empty. 

Now suppose that $p\alpha< z$. We show that the second element $(j+k,B)$ of the set $D_{(j,T)}'$ does not belong to any set $C_\beta'$, and hence indeed belongs to the class $D_{(j,T)}$. Let $\beta\in[\![0,N]\!]$ and suppose to the contrary that $(j+k,B)\in C_\beta'$. Since $j+k\in[\![0,2k-1]\!]$, the state $(j+k,B)$ must be either the first or the second element of the set $C_\beta'$. But since $B\ne T_0=T$, it has to be the second. If $p\beta< z$, then we obtain $j+k=\lfloor\frac{r}{2^{p\beta}}\rfloor+k2^{z-p\beta}\ge 2k$, a contradiction. Thus $p\beta\ge z$ and $j+k=\lfloor\frac{r}{2^{p\beta}}\rfloor+k$. But this implies that $\lfloor\frac{r}{2^{p\alpha}}\rfloor=j=\lfloor\frac{r}{2^{p\beta}}\rfloor$. Since $\beta> \alpha$, this means that $j=0$, a contradiction.
\end{proof}

\begin{corollary}
\label{cor:empty-D}
There are exactly $N-\lceil \frac{z}{p}\rceil$ empty classes among the $2k-1$ classes $D_{(j,X)}$.
\end{corollary}

\begin{proof}
By Proposition~\ref{prop:empty-D}, the $k$ classes $D_{(0,B)},\ldots,D_{(k-1,B)}$ are nonempty and we have to count the number of classes of the form $D_{(\lfloor\frac{r}{2^{p\alpha}}\rfloor,T)}$ with $p\alpha\ge z$ among the $k-1$ classes $D_{(1,T)},\ldots,D_{(k-1,T)}$. Note that by Lemma~\ref{lem:number-classes} and by definition of $N$, we have $N-\lceil\frac{z}{p}\rceil\in[\![0,k-1]\!]$. 

Equivalently, we have to count the elements $\alpha\in[\![\lceil \frac{z}{p}\rceil,N]\!]$ such that $\lfloor\frac{r}{2^{p\alpha}}\rfloor\ne 0$. Similarly to the proof of Proposition~\ref{prop:nonempty-C}, we consider two cases (also see Figure~\ref{fig:nonempty-classes}). If $N=\lceil \frac{z}{p}\rceil$ then there is no such $\alpha$ at all since $\lfloor\frac{r}{2^{pN}}\rfloor= 0$. If $N=R>\lceil \frac{z}{p}\rceil$, then the suitable $\alpha$ are exactly those in $[\![\lceil \frac{z}{p}\rceil,R-1]\!]$, and there are exactly $R-1-\lceil \frac{z}{p}\rceil+1=N-\lceil \frac{z}{p}\rceil$ of them. Hence the conclusion.
\end{proof}

\subsection{States of the same class are indistinguishable}
\label{sec:reduction1}

In order to prove that two states $(j,X)$ and $(j',X')$ of the automaton $\Pi \left( \A_{m,r,2^p} \times \A_{\T,2^p} \right)$ are indistinguishable, we have to prove that $L_{(j,X)}=L_{(j',X')}$. The general procedure that we use for proving that $L_{(j,X)}= L_{(j',X')}$ goes as follows. Pick some word $v\in A_{2^p}^*$ and let $n=|v|$ and $e=\val_{2^p}(v)$. 
By Lemma~\ref{lem:transitionsProd-proj}, the word $v$ is accepted from the state $(j,X)$ if and only if there exists some $d \in\N$ 
such that 
\[
	2^{pn}j+e = md+r \quad \andrm \quad X_d=T.
\]
Similarly, the word $v$ is accepted from the state $(j',X')$ if and only if there exists some $d' \in\N$ 
such that 
\[
	2^{pn}j+e = md'+r \quad \andrm \quad X_{d'}=T.
\]
But then, observe that there is only one possible pair of candidates for $d$ and $d'$: we necessarily have
\begin{equation}
\label{eq:d-d'}	
	d=\frac{2^{pn}j+e-r}{m}\quad \text{ and } \quad d'=\frac{2^{pn}j'+e-r}{m}.
\end{equation}
Therefore, proving that 
\[
	L_{(j,X)}= L_{(j',X')}
\] 
is equivalent to proving that for all $n\in\N$ and $e\in[\![0,2^{pn}{-}1]\!]$, we have
\[
(d\in\N \text{ and } X_d=T)\iff
(d'\in\N \text{ and } (X')_{d'}=T) 
\]
where $d$ and $d'$ are given by~\eqref{eq:d-d'}.
Moreover, note such $d$ and $d'$ are always greater than or equal to $-\frac{r}{m}$, hence they are greater than $-1$. Thus, provided that $d$ and $d'$ are integers, we know that they are necessary nonnegative. Similarly, thanks to Remark~\ref{rem:2^pn}, $d$ and $d'$ must be less than $2^{pn}$. For these reasons, in the forthcoming proofs (namely, in Lemmas~\ref{lem:0} and~\ref{lem:C-1}), we need to verify that $d,d'\in\Z$ but we don't need to check that $0\le d,d'<2^{pn}$. 
\medskip

Our first aim is to show that all states in the same class $D_{(j,X)}$ accept the same language. We start with a lemma that will be used several times. Note that this lemma does not only concern the classes $D_{(j,X)}$ since we can have $(j,X)=(0,T)$ in the statement.

\begin{lemma}
\label{lem:0}
Let $j\in[\![0,k-1]\!]$, $\ell\in[\![0,2^z-1]\!]$ and $X\in\{T,B\}$. For all $n\in\N$ such that $pn\ge z$, we have
\[
	L_{(j,X)}\cap (A_{2^p})^n
	=L_{(j+\ell k,X_\ell)}\cap (A_{2^p})^n.
\]
\end{lemma}

\begin{proof}
Let $n\in\N$ such that $pn\ge z$ and let $e\in[\![0,2^{pn}-1]\!]$. Set
\[
	d=\frac{2^{pn}j+e-r}{m}
	\quad \text{and}\quad 
	d'=\frac{2^{pn}(j+\ell k)+e-r}{m}.
\]
Following the procedure described above, we have to prove that $(d\in\N \andrm X_d=T) \iff (d'\in\N \andrm (X_\ell)_{d'}=T)$. Since $d'=d+\frac{2^{pn}\ell k}{m}=d+\ell 2^{pn-z}$ and since $pn\ge z$, $d$ is an integer if and only if so is $d'$. Moreover, 
\begin{equation}
\label{eq:lem2}
	d\le \frac{2^{pn}j+e}{m}<\frac{2^{pn}(j+1)}{m}\le\frac{2^{pn}k}{m}=2^{pn-z}.
\end{equation}
If $d,d'\in\N$ then $\rep_2(d')=\rep_2(\ell)0^{pn-z-|\rep_2(d)|}\rep_2(d)$, and hence $X_d=(X_\ell)_{d'}$. 
\end{proof}

\begin{proposition}
\label{prop:min-D}
Let $(j,X)\in([\![0,k-1]\!]\times \{T,B\})\setminus \{(0,T)\}$. Then any two states in $D_{(j,X)}$ accept the same language.
\end{proposition}

\begin{proof}
Let $\ell,\ell'\in[\![0,2^z-1]\!]$. It suffices to show that if $(j+\ell k,X_\ell)\in D_{(j,X)}$ then $L_{(j+\ell k,X_\ell)}\subseteq L_{(j+\ell' k,X_{\ell'})}$. Thus, suppose that $(j+\ell k,X_\ell)\notin \cup_{\alpha=0}^N C_\alpha$. Let $n\in\N$ and $e\in[\![0,2^{pn}-1]\!]$. Set $d=\frac{2^{pn}(j+\ell k)+e-r}{m}$ and assume that $d\in\N$ and $(X_\ell)_d=T$. Then $X_\ell=T_d$. If $pn<z$ then $\frac{r-e}{2^{pn}}=\left\lfloor\frac{r}{2^{pn}}\right\rfloor$ because  $\frac{r-e+dm}{2^{pn}}=j+k\ell$ is an integer, $m$ is divisible by $2^{pn}$ and $e\in[\![0,2^{pn}-1]\!]$. Therefore, if $pn<z$ then we get that
\[
	(j+\ell k,X_\ell)
	=\Big(\frac{r-e+dm}{2^{pn}},T_d\Big)
	=\Big(\left\lfloor\frac{r}{2^{pn}}\right\rfloor+d\frac{m}{2^{pn}},T_d\Big)
	\in C_n'
\]
which contradicts our assumption.  So $pn\ge z$ and the conclusion follows from Lemma~\ref{lem:0}.
\end{proof}

Note that the proof of Proposition~\ref{prop:min-D} shows that no word shorter than $\lfloor\frac{z}{p}\rfloor$ is accepted from a state of a class $D_{(j,X)}$. However, such words may be accepted from a state of one of the classes $C_\alpha$ (see Lemma~\ref{lem:suffixes} below).

Now  we turn to the classes $C_\alpha$. The proof is divided in several technical lemmas.

\begin{lemma}
\label{lem:C-1}
For every $\alpha\in[\![0,N]\!]$, any two states in $C_\alpha'$ accept the same words of length at least $\alpha$.
\end{lemma}

\begin{proof}
Let $\alpha\in[\![0,N]\!]$. First, we do the case $\alpha \le \frac{z}{p}$. By definition of the sets $C_\alpha'$, it suffices to show that for all $\ell\in[\![0,2^{p\alpha}{-}1]\!]$ and $n\ge \alpha$, we have $L_{( \lfloor\frac{r}{2^{p\alpha}}\rfloor,T)}\cap (A_{2^p})^n=L_{( \lfloor\frac{r}{2^{p\alpha}}\rfloor+\ell \frac{m}{2^{p\alpha}},T_\ell)}\cap (A_{2^p})^n$. Thus, let $\ell\in[\![0,2^{p\alpha}{-}1]\!]$, $n\ge \alpha$ and $e\in[\![0,2^{pn}{-}1]\!]$. Then set 
\[
	d=\frac{2^{pn}\lfloor\frac{r}{2^{p\alpha}}\rfloor+e-r}{m}
	\quad\text{and}\quad 
	d'=\frac{2^{pn}(\lfloor\frac{r}{2^{p\alpha}}\rfloor+\ell \frac{m}{2^{p\alpha}})+e-r}{m}.
\]
We have to prove that $(d\in\N 
\text{ and }T_d=T) \iff (d'\in\N
\text{ and }(T_\ell)_{d'}=T)$. Since $\llfloor\frac{r}{2^{p\alpha}}\rrfloor<\frac{m}{2^{p\alpha}}=k2^{z-p\alpha}$ and $z-p\alpha\ge 0$, we obtain that $\llfloor\frac{r}{2^{p\alpha}}\rrfloor+1\le \frac{m}{2^{p\alpha}}$. Then
\begin{equation}
\label{eqn:1}
 	d<\frac{2^{pn}\big(\llfloor\frac{r}{2^{p\alpha}}\rrfloor+1\big)}{m}
 	\le 2^{p(n-\alpha)}.
\end{equation}
Since $d'= d +\ell 2^{p(n-\alpha)}$ and $n\ge \alpha$, it follows that $d$ is an integer if and only if so is $d'$. Moreover, if both $d$ and $d'$ are in $\N$ then $\rep_2(d')=\rep_2(\ell)0^{p(n-\alpha)-|\rep_2(d)|}\rep_2(d)$, hence $T_d=(T_\ell)_{d'}$. 

Second, suppose that $\alpha >\frac{z}{p}$. In this case, we have to show that for all $\ell\in[\![0,2^z{-}1]\!]$ and $n\ge \alpha$, we have $L_{( \lfloor\frac{r}{2^{p\alpha}}\rfloor,T)}\cap (A_{2^p})^n = L_{( \lfloor\frac{r}{2^{p\alpha}}\rfloor+\ell k,T_\ell)}\cap (A_{2^p})^n$. Since $\llfloor\frac{r}{2^{p\alpha}}\rrfloor<\frac{m}{2^{p\alpha}}=k2^{z-p\alpha}< k$, the conclusion follows from Lemma~\ref{lem:0}.
\end{proof}





\begin{lemma}
\label{lem:C-2}
Let $\alpha\in[\![0,N]\!]$.
\begin{enumerate}
	\item If $p\alpha\le z$ then no state in $C_\alpha$ accepts any words of length $<\alpha$.
	\item If $p\alpha> z$ then no state in $C_\alpha$ accepts any words of length  $\le \llfloor\frac{z}{p}\rrfloor$.
\end{enumerate}
\end{lemma}

\begin{proof}
Let us prove item 1. Suppose that $p\alpha\le z$ and that there exists a word over $A_{2^p}$ of length $\beta<\alpha$ that is accepted from a state of the form $(\lfloor\frac{r}{2^{p\alpha}}\rfloor+\ell \frac{m}{2^{p\alpha}},T_\ell)$ with $\ell\in[\![0,2^{p\alpha}-1]\!]$, i.e.\ from a state in $C_\alpha'$. This means that there exists $e\in[\![0,2^{p\beta}-1]\!]$ such that if we set
\[
	d=\frac{2^{p\beta}(\lfloor\frac{r}{2^{p\alpha}}\rfloor+\ell \frac{m}{2^{p\alpha}})+e-r}{m},
\]
then $d\in\N$ and $(T_\ell)_d=T$. But then
\[
	\Big(\left\lfloor\frac{r}{2^{p\alpha}}\right\rfloor+\ell \frac{m}{2^{p\alpha}},T_\ell\Big)
	=\Big(\frac{r-e+dm}{2^{p\beta}},T_d\Big)
	=\Big(\left\lfloor\frac{r}{2^{p\beta}}\right\rfloor+d\frac{m}{2^{p\beta}},T_d\Big)\in C_\beta'.
\]
Since $\beta<\alpha$, the state $(\lfloor\frac{r}{2^{p\alpha}}\rfloor+\ell \frac{m}{2^{p\alpha}},T_\ell)$ does not belong to $C_\alpha$.

Now we prove item 2. Suppose that $p\alpha> z$ and that there exists a word over $A_{2^p}$ of length $\beta\le \llfloor\frac{z}{p}\rrfloor$ that is
accepted from a state of the form $(\lfloor\frac{r}{2^{p\alpha}}\rfloor+\ell k,T_\ell)$ with $\ell\in[\![0,2^z-1]\!]$, i.e.\ from a state in $C_\alpha'$. This means that there exists $e\in[\![0,2^{p\beta}-1]\!]$ such that if we set
\[
	d=\frac{2^{p\beta}(\lfloor\frac{r}{2^{p\alpha}}\rfloor+\ell k)+e-r}{m},
\]
then $d\in\N$ and $(T_\ell)_d=T$. But then
\[
	\Big(\left\lfloor\frac{r}{2^{p\alpha}}\right\rfloor+\ell k,T_\ell\Big)
	=\Big(\frac{r-e+dm}{2^{p\beta}},T_d\Big)
	=\Big(\left\lfloor\frac{r}{2^{p\beta}}\right\rfloor+d\frac{m}{2^{p\beta}},T_d\Big)\in C_\beta'.
\]
Therefore the state $(\lfloor\frac{r}{2^{p\alpha}}\rfloor+\ell k,T_\ell)$ does not belong to $C_\alpha$.
\end{proof}

\begin{lemma}
\label{lem:C-3}
If $N=\lceil\frac{z}{p}\rceil$, then no state in $C_N$ accepts any words of length $<N$.
\end{lemma}

\begin{proof}
Suppose that $N=\lceil\frac{z}{p}\rceil$ and that $v$ is a word over $A_{2^p}$ of length $\beta<N$ that is accepted from a state of the form $(\lfloor\frac{r}{2^{p\alpha}}\rfloor+\ell k,T_\ell)$ with $\ell\in[\![0,2^z-1]\!]$, i.e.\ from a state in $C_N'$. Set
\[
	e=\val_{2^p}(v) 
	\quad \text{and}\quad
	d=\frac{2^{p\beta}(\lfloor\frac{r}{2^{p\alpha}}\rfloor+\ell k)+e-r}{m}.
\]
Then $d\in\N$ and $(T_\ell)_d=T$. We get
\[
\left\lfloor\frac{r}{2^{p\alpha}}\right\rfloor+\ell k
=\frac{r-e+dm}{2^{p\beta}}
=\left\lfloor\frac{r}{2^{p\beta}}\right\rfloor+d\frac{m}{2^{p\beta}}
\]
and $T_\ell=T_d$. Since  $\beta<\frac{z}{p}$, this shows that the state $(\llfloor\frac{r}{2^{p\alpha}}\rrfloor+\ell k,T_\ell)=(\llfloor\frac{r}{2^{p\beta}}\rrfloor+d\frac{m}{2^{p\beta}},T_d)$ belongs to $C_\beta'$, and hence cannot belong to $C_N$.
\end{proof}

We are now ready to prove that two states belonging to any given class $C_\alpha$ are indistinguishable. 

\begin{proposition}
\label{prop:min-C}
For every $\alpha\in[\![0,N]\!]$, any two states in $C_\alpha$ accept the same language.
\end{proposition}

\begin{proof}
Let $\alpha\in[\![0,N]\!]$. From Lemma~\ref{lem:C-1}, it is enough to consider words of length smaller than $\alpha$ and from the first item of Lemma~\ref{lem:C-2}, we may suppose that $p\alpha> z$. If $N=\lceil\frac{z}{p}\rceil$, then we must have $\alpha=N$ and we are done thanks to Lemma~\ref{lem:C-3}. Thus, we may also assume that $N=R>\lceil\frac{z}{p}\rceil$. Under these assumptions, $\lfloor\frac{r}{2^{p\alpha}}\rfloor<k$ and the first state of $C_\alpha$ is $(\lfloor\frac{r}{2^{p\alpha}}\rfloor,T)$; see Figure~\ref{fig:nonempty-classes}.  Thus, we have to show that for all $\ell\in[\![0,2^z-1]\!]$ such that the state $(\lfloor\frac{r}{2^{p\alpha}}\rfloor+\ell k,T_\ell)$ indeed belongs to $C_\alpha$ and all $n<\alpha$, we have 
\[
	L_{(\lfloor\frac{r}{2^{p\alpha}}\rfloor,T)}\cap (A_{2^p})^n
	= L_{(\lfloor\frac{r}{2^{p\alpha}}\rfloor+\ell k,T_\ell)}\cap (A_{2^p})^n.
\]
If $pn<z$ then both languages are empty by the second item of Lemma~\ref{lem:C-2}. If $pn\ge z$ then  the equality follows from Lemma~\ref{lem:0}. 
\end{proof}

\subsection{States of different classes are distinguishable}
\label{sec:distinguishable}

In this section, we show that, in the projected automaton $\Pi \left( \A_{m,r,2^p} \times \A_{\T,2^p} \right)$, states belonging to different classes $C_\alpha$ or $D_{(j,X)}$ are pairwise distinguishable, that is, for any two such states, there exists a word which is accepted from exactly one of them.

The following lemma shows that the states in a set $C_\alpha'$ are exactly those states that leads to states in $C_{\alpha-1}'$ by reading the letter $r_{\alpha-1}$, where $0^{N-R}\rep_{2^p}(r)=r_{N-1}\cdots r_1r_0$.

\begin{lemma}
\label{lem:suffixes}
Let $0^{N-R}\rep_{2^p}(r)=r_{N-1}\cdots r_1r_0$ and let $\alpha\in[\![0,N]\!]$. Then
\[
	C_\alpha'=	
	\{(i,X)\in[\![0,m-1]\!]\times \{T,B\}\colon \delta_\times^\Pi((i,X),r_{\alpha-1}\cdots r_1r_0)=(r,T)\}.
\]
\end{lemma}

\begin{proof}
First, we consider the case where $p\alpha\le z$. Pick some $(i,X)\in C_\alpha'$. By definition, there exists $\ell\in[\![0,2^{p\alpha}-1]\!]$ such that $(i,X)=(\llfloor \frac{r}{2^{p\alpha}}\rrfloor+\ell\frac{m}{2^{p\alpha}},T_\ell)$. Observe that $\llfloor \frac{r}{2^{p\alpha}}\rrfloor=\val_{2^p}(r_{N-1}\cdots r_{\alpha+1}r_\alpha)$. Then
\[
	2^{p\alpha} \Big(\Big\lfloor \frac{r}{2^{p\alpha}}\Big\rfloor+\ell\frac{m}{2^{p\alpha}}\Big) 
	+ \val_{2^p}(r_{\alpha-1}\cdots r_1r_0) 
	= \ell m+r.
\]
Since $(T_\ell)_\ell=T$, we obtain from Lemma~\ref{lem:transitionsProd-proj} that $\delta_\times^\Pi((i,X),r_{\alpha-1}\cdots r_1r_0)=(r,T)$. Conversely, pick some state $(i,X)\in[\![0,m-1]\!]\times \{T,B\}$ such that $\delta_\times^\Pi((i,X),r_{\alpha-1}\cdots r_1r_0)=(r,T)$. Then there exists some $d\in[\![0,2^{p\alpha}-1]\!]$ such that 
\[
	2^{p\alpha} i+ \val_{2^p}(r_{\alpha-1}\cdots r_1r_0) 
	= m d +r
	\quad \andrm \quad X_d=T.
\]
We obtain
\[
i 	=\frac{1}{2^{p\alpha}}\Big(md+r-\val_{2^p}(r_{\alpha-1}\cdots r_1r_0)\Big) 
	=d\frac{m}{2^{p\alpha}}+ \Big\lfloor \frac{r}{2^{p\alpha}}\Big\rfloor.
\]
Observe that $X_d=T$ is equivalent to $X=T_d$. This proves that $(i,X)\in C_\alpha'$.

Second, we consider the case where $p\alpha> z$.
Pick some $(i,X)\in C_\alpha'$. There exists $\ell\in[\![0,2^z-1]\!]$ such that $(i,X)=(\llfloor \frac{r}{2^{p\alpha}}\rrfloor+\ell k,T_\ell)$. Then
\[
	2^{p\alpha} \Big(\Big\lfloor \frac{r}{2^{p\alpha}}\Big\rfloor+\ell k\Big) 
	+ \val_{2^p}(r_{\alpha-1}\cdots r_1r_0) 
	= \ell k2^{p\alpha}+r
	= \ell2^{p\alpha-z} m+r.
\]
Since $(T_\ell)_{\ell2^{p\alpha-z}}=(T_\ell)_\ell=T$, we obtain that $\delta_\times^\Pi((i,X),r_{\alpha-1}\cdots r_1r_0)=(r,T)$. Conversely, pick some state $(i,X)\in[\![0,m-1]\!]\times \{T,B\}$ such that $\delta_\times^\Pi((i,X),r_{\alpha-1}\cdots r_1r_0)=(r,T)$. Then there exists some $d\in[\![0,2^{p\alpha}-1]\!]$ such that 
\[
	2^{p\alpha} i+ \val_{2^p}(r_{\alpha-1}\cdots r_1r_0) 
	= m d +r
	\quad \andrm \quad X_d=T.
\]
From the first equality, we have $md = 2^{p\alpha} i- (r-\val_{2^p}(r_{\alpha-1}\cdots r_1r_0))=2^{p\alpha} i- 2^{p\alpha}\llfloor \frac{r}{2^{p\alpha}}\rrfloor$, hence
$kd = 2^{p\alpha-z} (i-\llfloor \frac{r}{2^{p\alpha}}\rrfloor)$.
Since $k$ is odd, $d$ must be a multiple of ${2^{p\alpha-z}}$. We obtain
\[
i =\frac{d}{2^{p\alpha-z}}k+ \Big\lfloor \frac{r}{2^{p\alpha}}\Big\rfloor
\]
and 
$X_{\frac{d}{2^{p\alpha-z}}}=X_d=T$. Since $\frac{d}{2^{p\alpha-z}}\in[\![0,2^z-1]\!]$,
we get that $(i,X)\in C_\alpha'$.
\end{proof}

\begin{proposition}
\label{prop:red-C}
For every $\alpha\in[\![0,N]\!]$, the class $C_\alpha$ is distinguishable from all the other classes.
\end{proposition}

\begin{proof}
First, we show that the classes $C_\alpha$ are distinguishable among them. From Proposition~\ref{prop:nonempty-C}, we know that these classes are all nonempty.
Let $\alpha,\beta\in[\![0,N]\!]$ such that $\alpha<\beta$ and let $(i,X)\in C_\alpha$ and $(j,Y)\in C_\beta$. We 
show that $L_{(i,X)}\ne L_{(j,Y)}$. By definition of the classes, the state $(i,X)$ belongs to $C_\alpha'$ and since $\alpha<\beta$, the state $(j,Y)$ does not belong to $C_\alpha'$. We get from Lemma~\ref{lem:suffixes} that the suffix $s$ of length $\alpha$ of the word $0^{N-R}\rep_{2^p}(r)$ is accepted from $(i,X)$ but not from $(j,Y)$. So $s\in L_{(i,X)}\setminus L_{(j,Y)}$.

Second, we show that the classes $C_\alpha$ are distinguishable from all the nonempty classes of the form $D_{(i,X)}$. Let $\alpha\in[\![0,N]\!]$. By definition, any state in a class $D_{(i,X)}$ cannot belong to $C_\alpha'$. Similarly to what precedes, the conclusion follows from Lemma~\ref{lem:suffixes}.
\end{proof}

It remains to show that the nonempty classes $D_{(j,X)}$ are distinguishable from each other.

\begin{proposition}
\label{prop:red-D}
Suppose that $k>1$ and let $(i,X),(j,Y)\in([\![0,k{-}1]\!]\times \{T,B\})\setminus\{(0,T)\}$ be distinct and such that the classes $D_{(i,X)}$ and $D_{(j,Y)}$ are both nonempty. Then $D_{(i,X)}$ and $D_{(j,Y)}$ are distinguishable.
\end{proposition}

\begin{proof}
We already know from the previous section that the states of $D_{(i,X)}$ (resp. $D_{(j,Y)}$) are indistinguishable. Therefore, it suffices to show that $L_{(i,X)}\ne L_{(j,Y)}$.

First, suppose that $i=j$. Then $X\ne Y$ by hypothesis and the states $(i,X)$ and $(j,Y)$ are disjoint by Lemma~\ref{lem:projdisj}. Since $\Pi \left(\A_{m,r,2^p} \times \A_{\T,2^p}  \right)$ is coaccessible by Proposition~\ref{prop:m-odd}, we obtain that the states $(i,X)$ and $(j,Y)$ are distinguishable.

Now suppose that $i\ne j$. By Lemma~\ref{lem:ij}, the word $w_i\rep_{2^p}(r)$ is accepted from $i$ in the automaton $\Pi(\A_{m,r,2^p})$ but is not accepted from $j$. Then, there exist a word $u_1$ of length $|w_i|$ and a word $u_2$ of length $R$ such that the word $(u_1,w_i)(u_2,\rep_{2^p}(r))$ is accepted from $i$ in the automaton $\A_{m,r,2^p}$ but is not accepted from $j$. By Lemma~\ref{lem:TM-product}, this word is accepted either from $(i,T)$ or from $(i,B)$ in the automaton $\A_{m,r,2^p}\times \A_{\T,2^p}$ 
but is not accepted neither from $(j,T)$ nor from $(j,B)$. Now, two cases are possible. 

First, suppose that $(u_1,w_i)(u_2,\rep_{2^p}(r))$ is accepted from $(i,X)$ in $\A_{m,r,2^p}\times \A_{\T,2^p}$. Then, in the projection $\Pi \left(\A_{m,r,2^p} \times \A_{\T,2^p}  \right)$, the word $w_i\rep_{2^p}(r)$ is accepted from $(i,X)$ but not from $(j,Y)$. Thus, the word $w_i\rep_{2^p}(r)$ distinguishes the states $(i,X)$ and $(j,Y)$. 

Second, suppose that $(u_1,w_i)(u_2,\rep_{2^p}(r))$ is accepted from $(i,\overline{X})$ in $\A_{m,r,2^p}\times \A_{\T,2^p}$. Let $(i',X')=\delta_\times((i,\overline{X}),(u_1,w_i))$. Then $\delta_\times((i',X'),(u_2,\rep_{2^p}(r))=(r,T)$. In particular $\delta_{m,r,2^p}^\Pi(i,w_i)=i'$, hence $i'=0$ by Lemma~\ref{lem:wi}. Now, by using Lemma~\ref{lem:TM-product} and Lemma~\ref{lem:1-m} successively, we obtain
\begin{align*}
	\delta_\times\big((i,X),(u_1,w_i)\rep_{2^p}(1,m)(u_2,\rep_{2^p}(r))\big)
	&=	\delta_\times\big((0,\overline{X'}),\rep_{2^p}(1,m)(u_2,\rep_{2^p}(r))\big) \\
	&= 	\delta_\times\big((0,X'),(u_2,\rep_{2^p}(r))\big) \\
	&= 	(r,T).
\end{align*} 
This shows that the word $w_i\rep_{2^p}(m)\rep_{2^p}(r)$ is accepted from $(i,X)$ in $\Pi \big(\A_{m,r,2^p} \times \A_{\T,2^p}  \big)$. From Lemmas~\ref{lem:ij} and~\ref{lem:projdisj-0}, this word cannot be accepted from $(j,Y)$, hence it distinguishes the states $(i,X)$ and $(j,Y)$.
\end{proof}

\subsection{The minimal automaton of $\val_{2^p}^{-1}(m \T+r)$.}
We are ready to construct the minimal automaton of $\val_{2^p}^{-1}(m \T+r)$. Since the states of $\Pi\left(\A_{m,r,2^p} \times \A_{\T,2^p} \right)$ that belong to the same class $C_\alpha$ or $D_{(j,X)}$ are indistinguishable, they can be glued together in order to define a new automaton $\mathcal{M}_{m,r,\T,2^p}$ that still accepts the same language. 

The formal definition of $\mathcal{M}_{m,r,\T,2^p}$ is as follows. The alphabet is $A_{2^p}$. The states are the classes $C_\alpha$ for $\alpha\in[\![0,N]\!]$  and the nonempty classes $D_{(j,X)}$ for $(j,X)\in\big([\![0,k{-}1]\!]\times \{T,B\}\big)\setminus\{(0,T)\}$. The class $C_R$ is the initial state and the only final state is the class $C_0$. Note that $(0,T)\in C_R$ and that $(r,T)\in C_0$. The transitions of $\mathcal{M}_{m,r,\T,2^p}$ are defined as follows: there is a transition labeled by a letter $a$ in $A_{2^p}$ from a class $J_1$ to a class $J_2$ if and only if in the automaton $\Pi \left( \A_{m,r,2^p}\times \A_{\T,2^p} \right)$, there is a transition labeled by $a$ from a state of $J_1$ to a state of $J_2$. 

\begin{example}
In Figure~\ref{fig:classe-6T}, the classes of $\Pi \left( \A_{6,2,4}\times \A_{\T,4} \right)$ are colored in white, blue, grey, yellow, fuchsia, orange and purple.
\begin{figure}[htb]
\centering
\begin{tikzpicture}[scale=0.8]
\tikzstyle{every node}=[shape=circle, fill=none, draw=black,
minimum size=30pt, inner sep=2pt]
\node[fill=violet](0T) at (0,0) {$0T$};
\node[fill=cyan](1T) at (3,0) {$1T$};
\node(2T) at (6,0) {$2T$};
\node[fill=yellow](3T) at (9,0) {$3T$};
\node[fill=magenta](4T) at (12,0) {$4T$};
\node[fill=orange](5T) at (15,0) {$5T$};
\node[fill=yellow](0B) at (0,-4.5) {$0B$};
\node[fill=magenta](1B) at (3,-4.5) {$1B$};
\node[fill=orange](2B) at (6,-4.5) {$2B$};
\node[fill=violet](3B) at (9,-4.5) {$3B$};
\node[fill=cyan](4B) at (12,-4.5) {$4B$};
\node[fill=gray](5B) at (15,-4.5) {$5B$};
\tikzstyle{every node}=[shape=circle, fill=none, draw=black,minimum size=25pt, inner sep=2pt]
\node at (6,0) {};

\tikzstyle{etiquettedebut}=[very near start,rectangle,fill=black!20]
\tikzstyle{etiquettemilieu}=[midway,rectangle,fill=black!20]
\tikzstyle{every path}=[color=black, line width=0.5 pt]
\tikzstyle{every node}=[shape=circle, minimum size=5pt, inner sep=2pt]
\draw [->] (-1.5,0) to node {} (0T); 
\draw [->] (0T) to [loop above] node [] {$0$} (0T);
\draw [green,->] (5T) to [loop above] node [] {} (5T);
\draw [->] (0B) to [loop below] node [] {} (0B);
\draw [green,->] (5B) to [loop below] node [] {} (5B);
\draw [blue,->] (0T) to [] node [above=-0.1] {$1$} (1T);
\draw [red,->] (0T) to [bend left=20] node [above=-0.1] {$2$} (2T);
\draw [green,->] (0T) to [bend left=25] node [above] {$3$} (3T);
\draw [->] (1T) to [bend left=20] node [] {} (4T);
\draw [blue,->] (1T) to [bend left=30] node [] {} (5T);
\draw [red,->] (1T) to [] node [] {} (0B);
\draw [green,->] (1T) to [bend left=15] node [] {} (1B);
\draw [->] (2T) to [bend left=15] node [] {} (2B);
\draw [blue,->] (2T) to [bend left=5] node [] {} (3B);
\draw [red,->] (2T) to [] node [] {} (4B);
\draw [green,->] (2T) to [] node [] {} (5B);
\draw [->] (3T) to [] node [] {} (0B);
\draw [blue,->] (3T) to [] node [] {} (1B);
\draw [red,->] (3T) to [bend left=5] node [] {} (2B);
\draw [green,->] (3T) to [bend left=15] node [] {} (3B);
\draw [->] (4T) to [bend left=15] node [] {} (4B);
\draw [blue,->] (4T) to [] node [] {} (5B);
\draw [red,->] (4T) to [bend right=35] node [] {} (0T);
\draw [green,->] (4T) to [bend right=25] node [] {} (1T);
\draw [->] (5T) to [bend right=25] node [] {} (2T);
\draw [blue,->] (5T) to [bend right=20] node [] {} (3T);
\draw [red,->] (5T) to [] node [] {} (4T);

\draw [blue,->] (0B) to [] node [] {} (1B);
\draw [red,->] (0B) to [bend right=20] node [] {} (2B);
\draw [green,->] (0B) to [bend right=25] node [] {} (3B);
\draw [->] (1B) to [bend right=20] node [] {} (4B);
\draw [blue,->] (1B) to [bend right=30] node [] {} (5B);
\draw [red,->] (1B) to [] node [] {} (0T);
\draw [green,->] (1B) to [bend left=15] node [] {} (1T);
\draw [->] (2B) to [bend left=15] node [] {} (2T);
\draw [blue,->] (2B) to [bend left=5] node [] {} (3T);
\draw [red,->] (2B) to [] node [] {} (4T);
\draw [green,->] (2B) to [] node [] {} (5T);
\draw [->] (3B) to [] node [] {} (0T);
\draw [blue,->] (3B) to [] node [] {} (1T);
\draw [red,->] (3B) to [bend left=5] node [] {} (2T);
\draw [green,->] (3B) to [bend left=15] node [] {} (3T);
\draw [->] (4B) to [bend left=15] node [] {} (4T);
\draw [blue,->] (4B) to [] node [] {} (5T);
\draw [red,->] (4B) to [bend left=35] node [] {} (0B);
\draw [green,->] (4B) to [bend left=25] node [] {} (1B);
\draw [->] (5B) to [bend left=25] node [] {} (2B);
\draw [blue,->] (5B) to [bend left=20] node [] {} (3B);
\draw [red,->] (5B) to [] node [] {} (4B);
\end{tikzpicture}
\caption{The classes of the automaton of $\Pi \left( \A_{6,2,4}\times \A_{\T,4} \right)$.}
\label{fig:classe-6T}
\end{figure}

Figure~\ref{fig:min-aut-6T} depicts the minimal automaton $\mathcal{M}_{6,2,\T,4}$ of $\val_4^{-1}(6\T+2)$, where states corresponding to the same color are glued together to form a single state. 
\begin{figure}[htb]
\centering
\begin{tikzpicture}[scale=0.8]
\tikzstyle{every node}=[shape=circle, fill=none, draw=black,
minimum size=30pt, inner sep=2pt]
\node[fill=violet](0T) at (0,0) {};
\node[fill=cyan](1T) at (4,0) {};
\node(2T) at (8,0) {};
\node[fill=yellow](0B) at (0,-4.5) {};
\node[fill=magenta](1B) at (4,-4.5) {};
\node[fill=orange](2B) at (8,-4.5) {};
\node[fill=gray](5B) at (12,-2.25) {};
\tikzstyle{every node}=[shape=circle, fill=none, draw=black,minimum size=25pt, inner sep=2pt]
\node at (8,0) {};

\tikzstyle{etiquettedebut}=[very near start,rectangle,fill=black!20]
\tikzstyle{etiquettemilieu}=[midway,rectangle,fill=black!20]
\tikzstyle{every path}=[color=black, line width=0.5 pt]
\tikzstyle{every node}=[shape=circle, minimum size=5pt, inner sep=2pt]
\draw [->] (-1.5,0) to node {} (0T); 
\draw [->] (0T) to [loop above] node [] {$0$} (0T);
\draw [->] (0B) to [loop below] node [] {} (0B);
\draw [green,->] (5B) to [loop above] node [] {} (5B);
\draw [blue,->] (0T) to [] node [above] {$1$} (1T);
\draw [red,->] (0T) to [bend left=25] node [above] {$2$} (2T);
\draw [green,->] (0T) to [bend left=15] node [right] {$3$} (0B);
\draw [->] (1T) to [bend left=25] node [] {} (1B);
\draw [blue,->] (1T) to [] node [] {} (2B);
\draw [red,->] (1T) to [] node [] {} (0B);
\draw [green,->] (1T) to [bend left=10] node [] {} (1B);
\draw [->] (2T) to [bend left=15] node [] {} (2B);
\draw [blue,->] (2T) to [bend right=45] node [] {} (0T);
\draw [red,->] (2T) to [] node [] {} (1T);
\draw [green,->] (2T) to [] node [] {} (5B);

\draw [blue,->] (0B) to [] node [] {} (1B);
\draw [red,->] (0B) to [bend right=25] node [] {} (2B);
\draw [green,->] (0B) to [bend left=15] node [] {} (0T);
\draw [->] (1B) to [bend left=25] node [] {} (1T);
\draw [blue,->] (1B) to [] node [] {} (5B);
\draw [red,->] (1B) to [] node [] {} (0T);
\draw [green,->] (1B) to [bend left=10] node [] {} (1T);
\draw [->] (2B) to [bend left=15] node [] {} (2T);
\draw [blue,->] (2B) to [bend left=45] node [] {} (0B);
\draw [red,->] (2B) to [] node [] {} (1B);
\draw [green,->] (2B) to [] node [] {} (2B);
\draw [->] (5B) to [] node [] {} (2B);
\draw [blue,->] (5B) to [bend left=10] node [] {} (0T);
\draw [red,->] (5B) to [] node [] {} (1T);
\end{tikzpicture}
\caption{The automaton $\mathcal{M}_{6,2,\T,4}$.}
\label{fig:min-aut-6T}
\end{figure}
\end{example}

\begin{theorem}
\label{thm:main2}
Let $p$ and $m$ be positive integers. The automaton $\mathcal{M}_{m,r,\T,2^p}$ is the minimal automaton of the language $\val_{2^p}^{-1}(m \T+r)$. 
\end{theorem}

\begin{proof}
By construction and by Propositions~\ref{prop:min-D} and~\ref{prop:min-C}, the language accepted by $\mathcal{M}_{m,r,\T,2^p}$ is $\val_{2^p}^{-1}(m \T+r)$. In order to see that $\mathcal{M}_{m,r,\T,2^p}$ is minimal, it suffices to prove that it is complete, reduced and accessible. The fact that $\mathcal{M}_{m,r,\T,2^p}$ is reduced follows from Propositions~\ref{prop:red-C} and \ref{prop:red-D}. We know from Proposition~\ref{prop:m-odd} that the automaton $\Pi\left(\A_{m,r,2^p} \times \A_{\T,2^p} \right)$ is complete and accessible, which in turn implies that $\mathcal{M}_{m,r,\T,2^p}$ is complete and accessible as well. 
\end{proof}


We are now ready to prove Theorem~\ref{thm:main}.

\begin{proof}[Proof of Theorem~\ref{thm:main}]
In view of Theorem~\ref{thm:main2}, it suffices to count the number of states of $\mathcal{M}_{m,r,\T,2^p}$. These states correspond to the nonempty classes  $C_\alpha$ and $D_{(j,X)}$. By Propositions~\ref{prop:nonempty-C} and~\ref{prop:empty-D}, there are exactly $(N+1)+\big(2k-1-(N-\llceil \frac{z}{p}\rrceil)\big)=2k+\llceil \frac{z}{p}\rrceil$ such classes. 
\end{proof}

\begin{example}
The minimal automaton of the language $\val_4^{-1} (6 \T+2)$  has $7$ states; see Figure~\ref{fig:min-aut-6T}. We can indeed compute that $2\cdot 3+\llceil \frac{1}{2} \rrceil = 7$.
\end{example}

\section{A decision procedure}
\label{sec:decision}

As an application of Theorem~\ref{thm:main}, we obtain the following decision procedure.

\begin{corollary}
Given any $2^p$-recognizable set $Y$ (via a finite automaton $\A$ recognizing it), it is decidable whether $Y=m\T+r$ for some $m\in\N$. The decision procedure can be run in time $O(N^2)$ where $N$ is the number of states of the given automaton $\A$.
\end{corollary}

\begin{proof}
Let $Y$ be a $2^p$-recognizable set given thanks to a complete DFA that accepts the language of the $2^p$-expansions of its elements. Let $N$ be its number of states. We can minimize and hence compute the state complexity $M$ of $Y$ (with respect to the base $2^p$) in time $O(N\log(N))$ \cite{Hopcroft:1971}. Let us decompose the possible multiples $m$ as $k2^z$ with $k$ odd. By Theorem~\ref{thm:main}, it is sufficient to test the equality between $Y$ and $m\T+r$ for the finitely many values of pairs $(k,z)$ such that $2k+\lceil\frac{z}{p}\rceil=M$. Since $M\le N$, the number of such tests is in $O(N)$. For each $m$ that has to be tested, we can directly use our description of the minimal automaton of $\val_{2^p}^{-1}(m\T+r)$ (this is Theorem~\ref{thm:main2}). This concludes the proof since the equality of two regular languages is decidable in linear time \cite{Hopcroft&Karp:2015}.
\end{proof}

\section{A direct description of the classes whenever $r=0$}
\label{sec:r=0}

In the conference paper \cite{Charlier&Cisternino&Massuir:2019}, we described the automaton $\mathcal{M}_{m,r,\T,2^p}$ in the particular case where $r=0$, i.e.\ for the exact multiples of $\T$. The construction was similar, but the way we build the classes of states was different. Therefore, we can give another description of the classes $C_\alpha$ and $D_{(j,X)}$ for $r=0$ which is easier than the descriptions from Definitions~\ref{def:C} and~\ref{def:D} in the sense that the classes are built in a direct way, without having to remove some states a posteriori. 

Note that if $r=0$ then $R=0$ and $N=\lceil\frac zp \rceil$.

\begin{corollary}
Suppose that $r=0$.
\begin{itemize}
\item We have $C_0=\{ (0,T)\}$.
\item For each $\alpha\in[\![1,N-1]\!]$, we have
\[
	C_\alpha=\bigcup_{\beta=\alpha p}^{\alpha p+p-1}					
				\{(k2^{z-\beta-1}+\ell k2^{z-\beta},B_\ell) 
					\colon \ell\in[\![0,2^\beta{-}1]\!]\}.
\]
\item We have 
\[
	C_N=\bigcup_{\beta=\left(\lceil\frac zp \rceil-1 \right)p}^{z{-}1}
				\{(k2^{z-\beta-1}+\ell k2^{z-\beta},B_\ell) 
					\colon \ell\in[\![0,2^\beta{-}1]\!]\}.
\]
\item For $(j,X) \in \big([\![0,k{-}1]\!]\times \{T,B\}\big)\setminus \{(0,T)\}$, we have
\[
	D_{(j,X)}=\{(j+\ell k,X_\ell) \colon \ell\in[\![0,2^z{-}1]\!]\}.
\] 
\end{itemize}
\end{corollary}

\begin{proof}
This is a consequence of Theorem~\ref{thm:main2} and \cite[Theorem~33]{Charlier&Cisternino&Massuir:2019}.
\end{proof}

\section{Replacing $\T$ by its complement $\overline{\T}$}
\label{sec:Tc}

If we are interested in the set $\overline{\T}=\N\setminus\T$ instead of $\T$, we can use the same construction that we described and studied for $\T$. We only have to exchange the final/non-final status of the states in the automaton $\A_{\T}$. In this section, we show that we may instead directly obtain the minimal automaton of the language $\val_{2^p}^{-1}(m \overline{\T}+r)$ from that of $\val_{2^p}^{-1}(m \T+r)$. 

\begin{example}
Let us push further our running example by considering now $\overline{\T}$ instead of $\T$. The classes of states are defined similarly by exchanging $T$ and $B$ everywhere. In Figure~\ref{fig:classe-6Tc}, we have depicted the classes of the corresponding projected product automaton, which we denote by $\Pi \big( \A_{6,2,4}\times \A_{\overline{\T},4} \big)$.
\begin{figure}[htb]
\centering
\begin{tikzpicture}[scale=0.8]
\tikzstyle{every node}=[shape=circle, fill=none, draw=black,
minimum size=30pt, inner sep=2pt]
\node[fill=yellow](0T) at (0,0) {$0T$};
\node[fill=magenta](1T) at (3,0) {$1T$};
\node[fill=orange](2T) at (6,0) {$2T$};
\node[fill=violet](3T) at (9,0) {$3T$};
\node[fill=cyan](4T) at (12,0) {$4T$};
\node[fill=gray](5T) at (15,0) {$5T$};
\node[fill=violet](0B) at (0,-4.5) {$0B$};
\node[fill=cyan](1B) at (3,-4.5) {$1B$};
\node(2B) at (6,-4.5) {$2B$};
\node[fill=yellow](3B) at (9,-4.5) {$3B$};
\node[fill=magenta](4B) at (12,-4.5) {$4B$};
\node[fill=orange](5B) at (15,-4.5) {$5B$};
\tikzstyle{every node}=[shape=circle, fill=none, draw=black,minimum size=25pt, inner sep=2pt]
\node at (6,-4.5) {};

\tikzstyle{etiquettedebut}=[very near start,rectangle,fill=black!20]
\tikzstyle{etiquettemilieu}=[midway,rectangle,fill=black!20]
\tikzstyle{every path}=[color=black, line width=0.5 pt]
\tikzstyle{every node}=[shape=circle, minimum size=5pt, inner sep=2pt]
\draw [->] (-1.5,0) to node {} (0T); 
\draw [->] (0T) to [loop above] node [] {$0$} (0T);
\draw [green,->] (5T) to [loop above] node [] {} (5T);
\draw [->] (0B) to [loop below] node [] {} (0B);
\draw [green,->] (5B) to [loop below] node [] {} (5B);
\draw [blue,->] (0T) to [] node [above=-0.1] {$1$} (1T);
\draw [red,->] (0T) to [bend left=20] node [above=-0.1] {$2$} (2T);
\draw [green,->] (0T) to [bend left=25] node [above] {$3$} (3T);
\draw [->] (1T) to [bend left=20] node [] {} (4T);
\draw [blue,->] (1T) to [bend left=30] node [] {} (5T);
\draw [red,->] (1T) to [] node [] {} (0B);
\draw [green,->] (1T) to [bend left=15] node [] {} (1B);
\draw [->] (2T) to [bend left=15] node [] {} (2B);
\draw [blue,->] (2T) to [bend left=5] node [] {} (3B);
\draw [red,->] (2T) to [] node [] {} (4B);
\draw [green,->] (2T) to [] node [] {} (5B);
\draw [->] (3T) to [] node [] {} (0B);
\draw [blue,->] (3T) to [] node [] {} (1B);
\draw [red,->] (3T) to [bend left=5] node [] {} (2B);
\draw [green,->] (3T) to [bend left=15] node [] {} (3B);
\draw [->] (4T) to [bend left=15] node [] {} (4B);
\draw [blue,->] (4T) to [] node [] {} (5B);
\draw [red,->] (4T) to [bend right=35] node [] {} (0T);
\draw [green,->] (4T) to [bend right=25] node [] {} (1T);
\draw [->] (5T) to [bend right=25] node [] {} (2T);
\draw [blue,->] (5T) to [bend right=20] node [] {} (3T);
\draw [red,->] (5T) to [] node [] {} (4T);

\draw [blue,->] (0B) to [] node [] {} (1B);
\draw [red,->] (0B) to [bend right=20] node [] {} (2B);
\draw [green,->] (0B) to [bend right=25] node [] {} (3B);
\draw [->] (1B) to [bend right=20] node [] {} (4B);
\draw [blue,->] (1B) to [bend right=30] node [] {} (5B);
\draw [red,->] (1B) to [] node [] {} (0T);
\draw [green,->] (1B) to [bend left=15] node [] {} (1T);
\draw [->] (2B) to [bend left=15] node [] {} (2T);
\draw [blue,->] (2B) to [bend left=5] node [] {} (3T);
\draw [red,->] (2B) to [] node [] {} (4T);
\draw [green,->] (2B) to [] node [] {} (5T);
\draw [->] (3B) to [] node [] {} (0T);
\draw [blue,->] (3B) to [] node [] {} (1T);
\draw [red,->] (3B) to [bend left=5] node [] {} (2T);
\draw [green,->] (3B) to [bend left=15] node [] {} (3T);
\draw [->] (4B) to [bend left=15] node [] {} (4T);
\draw [blue,->] (4B) to [] node [] {} (5T);
\draw [red,->] (4B) to [bend left=35] node [] {} (0B);
\draw [green,->] (4B) to [bend left=25] node [] {} (1B);
\draw [->] (5B) to [bend left=25] node [] {} (2B);
\draw [blue,->] (5B) to [bend left=20] node [] {} (3B);
\draw [red,->] (5B) to [] node [] {} (4B);
\end{tikzpicture}
\caption{The classes of the automaton of $\Pi\left(\A_{6,2,4}\times \A_{\overline{\T},4}\right)$.}
\label{fig:classe-6Tc}
\end{figure}
Figure~\ref{fig:min-aut-6Tc} depicts the minimal automaton $\mathcal{M}_{6,2,\overline{\T},4}$ of $\val_4^{-1}(6\overline{\T}+2)$, where states corresponding to the same color are glued together to form a single state. 
\begin{figure}[htb]
\centering
\begin{tikzpicture}[scale=0.8]
\tikzstyle{every node}=[shape=circle, fill=none, draw=black,
minimum size=30pt, inner sep=2pt]
\node[fill=violet](0T) at (0,-4.5) {};
\node[fill=cyan](1T) at (4,-4.5) {};
\node(2T) at (8,-4.5) {};
\node[fill=yellow](0B) at (0,0) {};
\node[fill=magenta](1B) at (4,0) {};
\node[fill=orange](2B) at (8,0) {};
\node[fill=gray](5B) at (12,-2.25) {};
\tikzstyle{every node}=[shape=circle, fill=none, draw=black,minimum size=25pt, inner sep=2pt]
\node at (8,-4.5) {};

\tikzstyle{etiquettedebut}=[very near start,rectangle,fill=black!20]
\tikzstyle{etiquettemilieu}=[midway,rectangle,fill=black!20]
\tikzstyle{every path}=[color=black, line width=0.5 pt]
\tikzstyle{every node}=[shape=circle, minimum size=5pt, inner sep=2pt]
\draw [->] (-1.5,0) to node {} (0B); 
\draw [->] (0T) to [loop below] node [] {} (0T);
\draw [->] (0B) to [loop above] node [] {$0$} (0B);
\draw [green,->] (5B) to [loop below] node [] {} (5B);
\draw [blue,->] (0T) to [] node [] {} (1T);
\draw [red,->] (0T) to [bend right=25] node [] {} (2T);
\draw [green,->] (0T) to [bend left=15] node [] {} (0B);
\draw [->] (1T) to [bend left=25] node [] {} (1B);
\draw [blue,->] (1T) to [] node [] {} (2B);
\draw [red,->] (1T) to [] node [] {} (0B);
\draw [green,->] (1T) to [bend left=10] node [] {} (1B);
\draw [->] (2T) to [bend left=15] node [] {} (2B);
\draw [blue,->] (2T) to [bend left=45] node [] {} (0T);
\draw [red,->] (2T) to [] node [] {} (1T);
\draw [green,->] (2T) to [] node [] {} (5B);

\draw [blue,->] (0B) to [] node [above] {$1$} (1B);
\draw [red,->] (0B) to [bend left=25] node [above] {$2$} (2B);
\draw [green,->] (0B) to [bend left=15] node [right] {$3$} (0T);
\draw [->] (1B) to [bend left=25] node [] {} (1T);
\draw [blue,->] (1B) to [] node [] {} (5B);
\draw [red,->] (1B) to [] node [] {} (0T);
\draw [green,->] (1B) to [bend left=10] node [] {} (1T);
\draw [->] (2B) to [bend left=15] node [] {} (2T);
\draw [blue,->] (2B) to [bend right=45] node [] {} (0B);
\draw [red,->] (2B) to [] node [] {} (1B);
\draw [green,->] (2B) to [] node [] {} (2B);
\draw [->] (5B) to [] node [] {} (2B);
\draw [blue,->] (5B) to [bend right=10] node [] {} (0T);
\draw [red,->] (5B) to [] node [] {} (1T);
\end{tikzpicture}
\caption{The automaton $\mathcal{M}_{6,2,\overline{\T},4}$.}
\label{fig:min-aut-6Tc}
\end{figure}
Since the classes of states have been modified but the edges are unchanged, the minimal automaton obtained by gluing the sets of the same classes together is not a symmetric version of the automaton $\mathcal{M}_{m,r,\T,2^p}$ we obtained starting from the set $\T$; compare Figures~\ref{fig:min-aut-6T} and~\ref{fig:min-aut-6Tc}. Nevertheless, observe that the automaton of Figure~\ref{fig:min-aut-6Tc} can be obtained from the one of Figure~\ref{fig:min-aut-6T} by replacing the initial state (in purple) by the yellow state. Also observe that, in the automaton of Figure~\ref{fig:min-aut-6T}, the yellow state is reached from the initial state by reading the word $\rep_4(6)=12$. This fact is always true and is proved in Proposition~\ref{prop:Tc}.
\end{example}

In the next proposition, we show that the minimal automaton of $\val^{-1}(m\overline{\T}+r)$ can be obtained directly from the minimal automaton of $\val_{2^p}^{-1}(m\T+r)$ by only moving the initial state. 

\begin{proposition}
\label{prop:Tc}
The minimal automaton of $\val_{2^p}^{-1}(m\overline{\T}+r)$ is obtained by replacing the initial state of the automaton $\mathcal{M}_{m,r,\T,2^p}$ by the state that is reached by reading $\rep_{2^p}(m)$ from the initial state.
\end{proposition}

\begin{proof}
Consider the automaton $\mathcal{M}_{m,r,\T,2^p}$. By construction, its states are sets of states (called classes) of the automaton $\Pi\left(\A_{m,r,2^p} \times \A_{\T,2^p} \right)$. By Lemma~\ref{lem:transitionsProd-proj}, for each $X\in\{T,B\}$, there is a path labeled by $\rep_{2^p}(m)$ going from $(0,X)$ to $(0,\overline{X})$ in $\Pi\left(\A_{m,r,2^p} \times \A_{\T,2^p} \right)$, and hence the same holds for the corresponding classes of states in $\mathcal{M}_{m,r,\T,2^p}$. 

First, let us show that the obtained automaton is again minimal. By only changing the initial state of any minimal DFA, we keep a DFA that is complete and reduced. Furthermore, the obtained DFA is still accessible since we have seen in the previous paragraph that there is a path from the class of $(0,B)$ to the class of $(0,T)$, which is precisely the initial state in $\mathcal{M}_{m,r,\T,2^p}$.

It remains to show that the language $L$ accepted from the class of $(0,B)$ in the automaton $\mathcal{M}_{m,r,\T,2^p}$ is equal to $\val_{2^p}^{-1}(m\overline{\T}+r)$. By construction, $L$ is equal to the language $L_{(0,B)}$ accepted from the state $(0,B)$ in the automaton $\Pi\left(\A_{m,r,2^p} \times \A_{\T,2^p} \right)$ and we already know that $L_{(0,T)}=\val_{2^p}^{-1}(m\T+r)$. 

Let $w\in A_{2^p}$. We know that $w\in L_{(0,B)}\iff \rep_{2^p}(m)w\in L_{(0,T)}$. Thus, it is sufficient to prove that $\val_{2^p}(w)\in m\overline{\T}+r\iff m2^{p|w|}+\val_{2^p}(w)\in m\T+r$.  In both cases, we must have that $\val_{2^p}(w)=mq+r$ with $q\in\N$. Since $q\le\val_{2^p}(w)<2^{p|w|}$, we have $\rep_2(2^{p|w|}+q)=10^{p|w|-|\rep_2(q)|}\rep_2(q)$. This shows that $q\in\overline{\T}\iff 2^{p|w|}+q\in\T$, hence the conclusion.
\end{proof}

\begin{corollary}
Let $m,p$ be positive integers and $r\in[\![0,m-1]\!]$. Then the state complexity of $m\overline{\T}+r$ with respect to the base $2^p$ is equal to $2k+\left\lceil \frac zp\right\rceil$ if $m=k2^z$ with $k$ odd.
\end{corollary}

\section{Conclusion and perspectives}
\label{sec:perspectives}

Our method is constructive and in principle, it may be applied to any $b$-recognizable set $X\subseteq\N$. However, in general, it is not the case that the product automaton $\A_{m,r,2^p} \times \A_{X,2^p}$ recognizing the bidimensional set $\{(n,mn+r)\colon n\in X\}$ is minimal. As an example, consider the $2$-recognizable set $X$ of powers of $2$: $X=\{2^n\colon n\in\N\}$. The product automaton $\A_{3,0,2} \times \A_{X,2}$ of our construction (for $m=3$, $r=0$ and $b=2$) has $6$ states but is clearly not minimal since it is easily checked that the automaton of Figure~\ref{fig:powers-2} is the trim minimal automaton recognizing the set $\{(2^n,3\cdot 2^n)\colon n\in  \N\}$.
\begin{figure}[htb]
\centering
\begin{tikzpicture}[scale=0.55]
\tikzstyle{every node}=[shape=circle, fill=none, draw=black,minimum size=20pt]
\node(1) at (0,0) {};
\node(2) at (5,0) {};
\node(3) at (10,0) {};
\tikzstyle{every node}=[shape=circle, fill=none, draw=black,minimum size=17pt]
\node at (10,0) {};

\tikzstyle{every path}=[color=black, line width=0.5 pt]
\tikzstyle{every node}=[shape=circle, minimum size=5pt, inner sep=2pt]

\draw [->] (-2,0) to node {} (1); 
\draw [->] (1) to [loop above] node [above=-0.3cm] {$(0,0)$} (1);
\draw [->] (3) to [loop above] node [above=-0.3cm] {$(0,0)$} (3);
\draw [->] (1) to [] node [above=-0.2cm] {$(0,1)$} (2);
\draw [->] (2) to [] node [above=-0.2cm] {$(1,1)$} (3);
\end{tikzpicture}
\caption{Minimal automaton recognizing the set $\{(2^n,3\cdot 2^n)\colon n\in\N\}$.}
\label{fig:powers-2}
\end{figure}
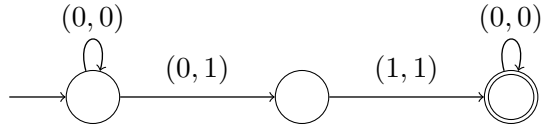
This illustrates that, in general, the minimization procedure is not only needed in the final projection $\Pi \left( \A_{m,r,2^p} \times \A_{X,2^p} \right)$ as is the case in the present work.

Nevertheless, we conjecture that the phenomenon described in this work for the Thue-Morse set also appears for all $b$-recognizable sets of the form 
\[	
	X_{b,c,M,R}=\{n\in\N\colon |\rep_b(n)|_c\equiv R\bmod M\}
\]
where $b$ is an integer base, $c$ is any digit in $A_b$, $M$ is an integer greater than or equal to $2$ and $R$ is any possible remainder in $[\![0,M-1]\!]$. More precisely, we conjecture that whenever the base $b$ is a prime power, i.e.\ $b=q^p$ for some prime $q$, then the state complexity of $mX_{b,c,M,R}+r$ is given by the formula $Mk+\lceil \frac{z}{p}\rceil$ where $k$ is the part of the multiple $m$ that is prime to the base $b$, i.e.\ $m=kq^z$ with $\gcd(k,q)=1$. Note that the set $\overline{\T}$ is of this form: $\overline{\T}=\{n\in\N\colon |\rep_2(n)|_1\equiv 0\bmod 2\}$.
 
Another potential future research direction in the continuation of the present work is to consider automata reading the expansions of numbers with least significant digit first. Both reading directions are relevant to different problems. For example, it is easier to compute addition thanks to an automaton reading expansions from “right to left” than from “left to right”. On the opposite, if we have in mind to generalize our problems to $b$-recognizable sets of real numbers (see for instance \cite{Boigelot&Rassart&Wolper:1998,Charlier:2018,Charlier&Leroy&Rigo:2015}), then the relevant reading direction is the one with most significant digit first. 
Further, there is no intrinsic reason why the state complexity from “left to right” should be the same as (or even close to) that obtained from “right to left” since in general, it is well known that the state complexity of an arbitrary language can greatly differ from that of its reversed language.

\section{Acknowledgment}
Célia Cisternino is supported by the FNRS Research Fellow grant 1.A.564.19F.

\bibliographystyle{abbrv}
\bibliography{TMmultiples}
\label{sec:biblio}

\end{document}